\documentclass[11pt]{article}

\usepackage[english]{babel}
\usepackage{amssymb, amsthm, amsmath}
\usepackage{mathtools}
\usepackage{graphicx}
\usepackage{xcolor}
\usepackage{stackrel}
\usepackage{adjustbox}
\usepackage{tabularx}
\usepackage{multicol}
\usepackage{multirow}
\usepackage[top=3cm, bottom=3.5cm, right=3cm, left=3cm]{geometry}

\newtheorem{theorem}{Theorem}
\newtheorem{proposition}{Proposition}
\newtheorem{notation}{Notation}
\newtheorem{example}{Example}
\newtheorem{remark}{Remark}

\newcommand{\B}{\mathbb{B}}
\newcommand{\N}{\mathbb{N}}

\newcommand{\ie}{\emph{i.e. }}
\newcommand{\fp}{\mathrm{fp}}
\newcommand{\fpzero}{(0,\dots,0)} 
\newcommand{\fpone}{(1,\dots,1)}
\newcommand{\F}{\mathrm{F}}
\newcommand{\FP}{\mathrm{FP}}
\newcommand{\LC}{\mathrm{LC}}
\newcommand{\stab}{\mathrm{s}}
\newcommand{\cy}{\mathrm{c}}
\newcommand{\lcm}{\mathrm{lcm}}
\newcommand{\loopr}{{}~\rotatebox{90}{$\circlearrowleft$}{}}

\newcommand{\gcI}{\mathrm{cI}}
\newcommand{\gcII}{\mathrm{cII}}
\newcommand{\gcro}{\mathrm{cro}}
\newcommand{\gN}{\mathrm{N}}

\newcommand{\gAPI}{\mathrm{AP1}}

\newcommand{\gBFU}{\mathrm{BFU}}
\newcommand{\gAG}{\mathrm{AG}}
\newcommand{\gAPIII}{\mathrm{AP3}}
\newcommand{\gPI}{\mathrm{PI}}

\begin{document}

\renewcommand{\floatpagefraction}{.95}

\setlength{\parindent}{0pt}

\title{Attractor landscapes in Boolean networks with firing memory
}
  
\author{Eric Goles$^{1}$, Fabiola Lobos$^{1}$, Gonzalo A. Ruz$^{1,2}$, Sylvain Sen{\'e}$^3$\\[2mm]
  {\small $^1~$Facultad de Ingener{\'i}a y Ciencias, Universidad Adolfo Iba{\~n}ez, 
			Santiago, Chile}\\
  {\small $^2~$Center of Applied Ecology and Sustainability (CAPES), Santiago, Chile}\\
  {\small $^3~$Aix-Marseille Univ., Toulon Univ., CNRS, LIS UMR7020, Marseille, France}
}

\date{}
  
\maketitle

\begin{abstract}
	In this paper we study the dynamical behavior of Boolean networks with firing 	
	memory, namely Boolean networks whose vertices are updated synchronously 
	depending on their proper Boolean local transition functions so that each vertex
	remains at its firing state a finite number of steps.	 We prove in particular that 
	these networks have the same computational power than the classical ones, \ie any 
	Boolean network with firing memory composed of $m$ vertices can be simulated by a 
	Boolean network by adding vertices. We also prove general results on specific 
	classes of networks. For instance, we show that the existence of at least one 
	delay greater than $1$ in disjunctive networks makes such networks have only 
	fixed points as attractors. Moreover, for arbitrary networks composed of two 
	vertices, we characterize the delay phase space, \ie the delay values such that 
	networks admits limit cycles or fixed points. Finally, we analyze two classical 
	biological models by introducing delays: the model of the immune control of the 
	$\lambda$-phage and that of the genetic control of the floral morphogenesis of 
	the plant \emph{Arabidopsis thaliana}.
\end{abstract}

\section{Introduction}
\label{sec:intro}

In the context of gene regulation modeling, the choice of the methodology highly 
depends on the nature of the underlying real system and on the objective of the 
study, that can be oriented towards quantitative or qualitative analysis of the 
dynamical behaviors of the networks. From the qualitative point of view, Boolean 
networks (BNs) are one of the simplest model and for more than forty years, they have 
been used to analyze and understand several biological phenomena. Notably, several 
BN models of real biological systems have become popular: the immunity 
control network of bacteriophage $\lambda$~\cite{Thieffry1995}, the floral 
morphogenesis network of \textit{Arabidopsis thaliana}~\cite{Mendoza1998}, the 
fission yeast cell-cycle network~\cite{Davidich2008}, the budding yeast cell-cycle 
network~\cite{Li2004}, the mammalian cell-cycle network~\cite{Faure2006}, the 
\textit{p53-mdm2} network~\cite{Choi2012}, and the blood cancer large granular 
lymphocyte (T-LGL) leukemia network~\cite{Zhang2008}.

Introduced by Kauffman at the end of the 1960's~\cite{Kauffman1969} by generalizing 
the classical formal neural networks of McCulloch and Pitts~\cite{McCulloch1943}, 
this model consists in a network where the vertices represent genes that can be 
expressed (or active, \ie vertex value $1$) or not (inactive, \ie vertex value $0$), 
and the edges represent regulatory relations between the genes. The dynamics of a 
network is then given by a set of Boolean functions, one for each vertex. Starting 
from any of the $2^n$ possible configurations (a configuration being a vector of 
$\B^n = \{0,1\}^n$), for a network composed of $n$ vertices, the dynamics of the 
network eventually converges towards ordered sets of recurrent configurations that 
repeat endlessly and periodically which we classically call attractors. When an 
attractor is composed of one configuration, it is called a fixed point; when it is 
composed of at least two configurations, we call it a limit cycle. Attractors are 
particularly relevant in the context of biological modeling because they are used to 
represent differentiated cellular types or tissues (in the case of fixed points) and 
biological rhythms or oscillations (in the case of limit cycles).

One of the characteristics of BNs is that they are associated with an update mode 
that defines the way vertices update their states along time. The parallel mode in 
which all the vertices are updated at each time step is canonical (\ie it is 
directly derived from the network definition) and belongs to the class of 
block-sequential update modes~\cite{Demongeot2008,Goles2008,Robert1986}. 
Block-sequential modes are deterministic and periodic and are defined by ordered 
partitions of the set of vertices. Another classical approach in the domain is to 
consider non-deterministic (and non-stochastic) update modes like the asynchronous 
one~\cite{Remy2003,Richard2007,Thomas1973} (stochastic asynchronicity, however, has 
been well studied in the context of cellular
automata~\cite{Dennunzio2013,Fates2014,Fates2005,Regnault2009}). Numerous 
studies have focused on the influence of update modes on the dynamical behaviors. 
From the theoretical point of view, among the most impacting analyses 
are~\cite{Aracena2009,Aracena2011,Goles2010} in the context of deterministic modes 
and~\cite{Noual2017} in that of non-deterministic ones. In both of these, the very 
relevance comes from the fact that the authors succeeded in explaining the influence 
of update modes on the dynamics of BNs by relating it to their static structures. 
From the applied point of view, the dynamical behaviors of many biological networks 
with different update schemes have been 
studied~\cite{Demongeot2010,Goles2013,Mendoza1999,Ruz2010,Ruz2014}.
This manner of studying the dynamical behaviors of biological networks is desirable 
when searching for biologically meaningful updating modes. However, as a matter of 
fact, although it is deeply interesting and relevant from formal points of view like 
mathematical and computational ones, this manner that consists in studying biological 
networks by considering as much updating schemes as possible is rather tedious (due to 
the infinite number of updating schemes, an updating scheme being defined from a 
general point of view as a function associating any subset of nodes with each time step 
of $\mathbb{N}$, \ie an infinite sequence of subsets of nodes) when the objective is 
fixed on the biological matter. Another approach that allows adding asynchronicity is 
based on the concept of delay. In the context of discrete modeling of biological 
regulation networks, among the first who have introduced delays is certainly 
Thomas~\cite{Thomas1991,Thomas1988,Thomas1995} whose works have been followed by many 
other in different 
frameworks~\cite{Ahmad2008,Bernot2004,Fromentin2010,Ren2008,Ribeiro2014}. Here, we 
make choice using a distinct approach based on considering BNs with memory, as the 
model studied in~\cite{Graudenzi2011a,Graudenzi2011b} that was initially developed by 
Graudenzi and Serra under the name of \emph{gene protein Boolean networks} 
(GPBNs)~\cite{Graudenzi2010}. As this name suggests, in this model, each vertex of 
the classical Boolean network is decoupled into both a gene vertex and a protein 
vertex, so that each pair of such vertices is associated with a decay time that acts 
as a memory standing for the number of steps during which the protein vertex remains 
active. 

In the seminal papers~\cite{Graudenzi2011a,Graudenzi2011b}, the authors focus 
on the provision of the memory effect due the addition of decay times. In 
particular, thanks to numerical simulations, they highlight very interesting 
properties: the memory effect significantly affects the robustness of the 
computational model itself against state perturbations, with respect to the classical 
model of random BNs; the more the maximum decay time value, the less the network admits 
asymptotic degrees of freedoms, \ie attractors; higher values of the maximum decay time 
results in longer limit cycles associated to attraction basins that are more ordered 
than in the case of (random) BNs. 

From this, we are convinced that this model deserves to be deeply studied, from both 
theoretical and applied points of view. That is what we propose to do in this paper, by 
following a constructive approach. Indeed, we will see that the GPBN model proposed by 
Graudenzi et al. is not more powerful than that of classical BNs from a strictly 
computational standpoint. Doing so, we will develop another equivalent intermediate 
representation merging gene and protein vertices that simplify substantially the phase 
space. This representation will be called \emph{Memory Boolean networks} (MBNs). We 
will also focus on specific classes of networks and pay particular attention to two 
genes networks which, despite their small size, allow acquiring much knowledge about 
the model. In addition, under a biological context, network traditionally are small, 
for example: Quorum-sensing systems in the plant growth-promoting bacterium ($5$ 
nodes)~\cite{Zuniga2017}, \textit{lac} operon in \textit{Escherichia coli} ($10$ 
nodes, which can be even further reduced to $3$ nodes)~\cite{Veliz2011}, 
oxidative stress response ($6$ nodes)~\cite{Leifeld2018}. Furthermore, 
there are examples where no prior knowledge (key genes) is available, and therefore, 
key genes cannot be selected from the hundreds or thousands of genes beforehand. In 
these cases, small Boolean networks have been inferred, where the nodes are metagenes 
(a group of genes that have similar co-expression patterns) identified via clustering 
in an earlier stage of the analysis, for example, the network of \textit{Arabidopsis 
thaliana} saline stress response ($12$ meta genes nodes, originally $569$ genes that 
were differentially regulated due to salt exposure)~\cite{Ruz2015}. 
Eventually, a pertinent constructive track initiated by Alon et 
al.~\cite{Alon2003,Alon2002,Alon2004} to achieve a better understanding of genetic 
networks consists in viewing them as compositions of small regulation motifs of $2$ or 
$3$ nodes (considered as ``building blocks of complex networks'') that deserve to be 
studied \textit{per se} before tackling their compositions. As a consequence, 
theoretical analyses of small networks are of interest in the context of modeling. 

This theoretical part will be followed by applications to two real biological systems: 
the immune control of the $\lambda$-phage and the genetic control of the floral 
morphogenesis of the plant \emph{Arabidopsis thaliana}.

\section{The models}
\label{sec:models}

\subsection{Definitions and notations}
\label{sec:models_def}

\subsubsection{Boolean networks (BNs)}
\label{sec:models_def_bns}

A BN $F$ of size $n$, \ie composed of $n$ genes, is a collection of $n$ Boolean local 
transition functions such that $F = (f_i: \B^n \to \B, f_i(x) \mapsto x_i)_{i \in 
\{1, \dots, n\}}$, where $x$ denotes a \emph{configuration} of $F$, and $x_i$ denotes 
the state of gene $i$. In a function $f_i$, consider it being minimal, if there is a 
positive (resp. negative) literal, for instance $x_j$ (resp. $\neg x_j$), this means 
that gene $j$ tends to activate (resp. inhibit) gene $i$. In other terms, the state 
of $i$ tends to mimic (resp. negate) that of $j$. From this can be easily derived a 
digraph $G = (V,E)$ where the vertex set is $V = \{1, \dots, n\}$ and where 
$E = \{(j, s, i)\ |\ s = + \text{ (resp. } s = - \text{) if } x_j \text{ (resp. } 
\neg x_j \text{) appears  in the definition of } f_i\}$. Such a graph $G$ is called 
the \emph{interaction graph} of $F$. As in Kauffman's seminal 
work~\cite{Kauffman1969}, let us consider for now on that BNs evolve in such a way 
that every gene updates its (expression) state at each time step, \ie in parallel. In 
this specific framework, the (global) dynamics of a BN $F$ is simply given by 
$\forall x \in \B^n,\ F(x) = (f_1(x), f_2(x), \dots, f_{n}(x))$, and can be written 
(by emphasing time steps): $\forall i \in V, \forall t \in \N,\ x_i(t+1) = 
f_i(x(t))$. Such a dynamics can be represented by its \emph{transition graph} that is 
the digraph $\mathcal{G} = (\B^n, F)$ (see Figure~\ref{fig_BN_example}).
\begin{figure}[t!]
	\begin{center}
		\begin{minipage}{.36\textwidth}
			\centerline{\scalebox{.85}{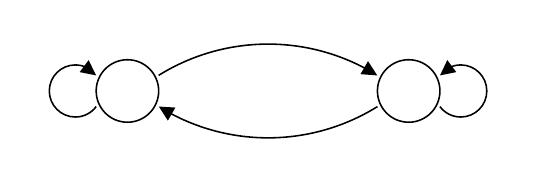}}
		\end{minipage}
		\quad
		\begin{minipage}{.275\textwidth}
			\centerline{\setlength{\tabcolsep}{2pt}\begin{tabular}{|cc|cc|}
				\hline
				$x_1$ & $x_2$ & $f_1$ & $f_2$\\
				\hline\hline
				0 & 0 & 0 & 0\\
				0 & 1 & 0 & 0\\
				1 & 0 & 1 & 1\\
				1 & 1 & 0 & 0\\
				\hline
			\end{tabular}}
		\end{minipage}
		\quad
		\begin{minipage}{.235\textwidth}
			\scalebox{.85}{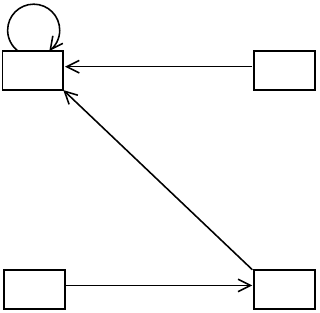}
		\end{minipage}\smallskip
		
		\begin{minipage}{.36\textwidth}
			\centerline{\scalebox{.85}{(a)}}
		\end{minipage}
		\quad
		\begin{minipage}{.275\textwidth}
			\centerline{\scalebox{.85}{(b)}}
		\end{minipage}
		\quad		
		\begin{minipage}{.235\textwidth}
			\centerline{\scalebox{.85}{(c)}}
		\end{minipage}\medskip
	\end{center}
	\caption{(a) Interaction graph, (b) truth tables of its local transition 
		functions and (c) transition graph of the BN composed of two genes $1$ 
		and $2$, defined by the local transition functions $f_1(x) = f_2(x) = 
		x_1 \land \neg x_2$. The transition graph shows that this BN admits 
		only one attractor, fixed point $(0,0)$.}
	\label{fig_BN_example}
\end{figure}
This graph represents more precisely the trajectories of all configurations towards 
attractors that are either fixed points or limit cycles as explained in the 
introduction.

\subsubsection{Gene protein Boolean networks (GPBNs)}
\label{sec_models_def_gpbns}

\noindent GPBNs were presented in~\cite{Graudenzi2010}. A GPBN $F$ can be viewed 
similarly to a BN by its interaction graph $G = (V,E)$, where each vertex of $V$ is 
decoupled into a gene and its associated protein. So, each gene of the network is 
strictly linked to a unique and specific protein. The vertex set is defined as $V = 
\{G_1, \dots, G_N, P_1, \dots, P_N\}$ and the Boolean local transition functions are 
given by $(f_{G_i}, f_{P_i}: \B^{N} \to \B)_{i \in \{1, \dots, N\}}$. Let us consider 
configuration $x = (x_{G_1}, \ldots,  x_{G_N}, x_{P_1}, \ldots, x_{P_N})$. If 
$x_{G_i} = 1$ (resp. $0$) then it means that gene $G_i$ is expressed or active 
(resp. unexpressed or inactive), and if $x_{P_i} = 1$ (resp. $0$), it means that 
protein $P_i$ is present (resp. absent) in the underlying cell. Every protein $P_i$, 
with $1 \leq i\leq N$, is associated with a \emph{decay time} $dt_i \in \N$. This 
decay time $dt_i$ defines the number of time steps during which $P_i$ remains present 
in the cell after having been produced by the punctual expression of gene $G_i$. 
Moreover, in this model, a delay of one time step is considered between a gene 
punctual expression and a protein for it to be considered as present. 

Formally, the global dynamics of a GPBN $F$ is defined as $\forall x \in \B^n\text{, 
with } n = 2N,\ F(x) = (f_{G_1}(x), \ldots, f_{G_N}(x), f_{P_1}(x), \ldots, 
f_{P_N}(x))$, where $\forall i \in \{1, \ldots, N\}$ and for any time step $t \in 
\N$:
\begin{multline*}
	x_{G_i}(t+1) = f_{G_i}(x(t))\quad \text{and}\quad 
	x_{P_i}(t+1) = \begin{cases}
		1 & \text{if } \Delta_{i}(t+1) \geq 1\\
		x_{G_i}(t) & \text{if } \Delta_{i}(t+1) = 0
	\end{cases}\text{,}
\end{multline*}
with:
\begin{equation*}
	\left\lbrace
		\begin{array}{rl}
			\Delta_i(0) & = \begin{cases}
				0 & \text{if } x_{P_i}(0) = 0\\
				\alpha \in \{1, \ldots, dt_i\} & \text{if } x_{P_i}(0) = 1
			\end{cases}\\
			\Delta_i(t+1) & = \begin{cases}
				0 & \text{if } x_{G_i}(t) = 0 \land \Delta_i(t) = 0\\
				\Delta_i(t)-1 & \text{if } x_{G_i}(t) = 0 \land \Delta_i(t) > 0\\
				dt_i & \text{if } x_{G_i}(t) = 1\\
			\end{cases}
		\end{array}
	\right.\text{.}
\end{equation*}

\subsubsection{Memory Boolean networks (MBNs)}
\label{sec_models_def_mbns}

In this paper, we propose a new model, that of MBNs. A MBN is defined by a digraph 
$G = (V,E)$, with $V = \{1, \dots, n\}$ and the Boolean local transition functions 
$(f_i: \B^n \to \B)_{i \in \{1, \dots, n\}}$. To this model is added a vector of 
delays $dt \in (\N \setminus \{0\})^n$ such that a firing vertex $i$ will remain at 
state $1$ during $dt_i$ time steps. Given an initial condition $x(0) \in \{0,1\}^n$ 
we consider the delays of each vertex so that $\Delta_i(0) = 0$ if $x_i(0) 
= 0$ and $\Delta_i(0) \in \{1, \dots, dt_i\}$ if $x_i(0) = 1$. More formally, a 
MBN is defined as $\forall x \in \B^n,\ F(x) = (f_1(x), f_2(x), \dots, f_{n}(x))$, 
where the update is:
\begin{equation*}
	x_i(t+1) = \begin{cases}
		1 & \text{if } \Delta_i(t+1) \geq 1\\
		f_i(x(t)) & \text{if } \Delta_i(t+1) = 0
	\end{cases}\text{,}
\end{equation*}	
and the delays are:
\begin{equation*}
	\Delta_i(t+1) = \begin{cases}
		0 & \text{if } f_i(x(t)) = 0 \text{ and } \Delta_i(t) = 0\\
		\Delta_i(t)-1 & \text{if } f_i(x(t)) = 0 \text{ and } \Delta_i(t) > 0\\
		dt_i & \text{if } f_i(x(t)) = 1\\
	\end{cases}\text{.}
\end{equation*}

In the sequel, we will use an abuse of notation for not burdening the reading. Rather 
than decoupling the state $x_i \in \B$ of a vertex $i$ from its associated 
delay $dt_i \in \N$, we will simply change the notation into $x_i \in \{0, \dots, 
dt_i\}$ so that, for all $i \in V$, if $x_i = 1$ and $dt_i = 2$, we will usually 
write $x_i = 2$. Due to this abuse of notation, if we consider for instance a MBN of 
size $3$ such that $dt = (2,1,1)$, configuration $x$ at time step $t$ denoted by 
$x(t) = (2,0,1)$ stands for the Boolean configuration $(1,0,1)$ (with $dt = (1,1,1)$) 
in the model such as it has been formally defined above.

\subsection{Equivalence(s) between BNs, GPBNs, MBNs}
\label{sec_models_equiv}

Definitions of BNs, GPBNs and MBNs above emphasize that these models match on many 
aspects. Here we will show that these models are equivalent in the following 
sense: given a GPBN $F$, it is always possible to build a MBN $F'$ and a BN 
$\tilde{F'}$ such that $F$, $F'$ and $\tilde{F'}$ admit equivalent asymptotic 
behaviors, in terms of type and number of attractors (of course, the attractors are 
not exactly composed of the same recurrent configurations because of the compression 
induced by the construction).
\begin{figure}[t!]
	\begin{center}
		\begin{minipage}{.175\textwidth}
			\centerline{\scalebox{.85}{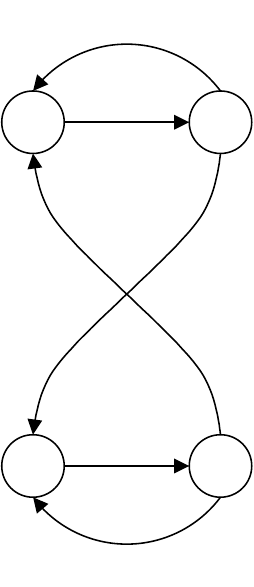}}
		\end{minipage}
		\quad
		\begin{minipage}{.4\textwidth}
			\centerline{\footnotesize\setlength{\tabcolsep}{1pt}
				\begin{tabular}{|cccc|cccc|}
					\hline
					$x_{P_1}$ & $x_{P_2}$ & $x_{G_1}$ & $x_{G_2}$ & 
					$f_{P_1} $ & $f_{P_2}$ & $f_{G_1}$ & $f_{G_2}$\\
					\hline
					\hline
					0 & 0 & 0 & 0 & 0 & 0 & 0 & 0 \\
					0 & 1 & 0 & 0 & 0 & 0 & 0 & 0 \\
					1 & 0 & 0 & 0 & 0 & 0 & 1 & 1 \\
					1 & 1 & 0 & 0 & 0 & 0 & 0 & 0 \\
					2 & 0 & 0 & 0 & 1 & 0 & 1 & 1 \\
					2 & 1 & 0 & 0 & 1 & 0 & 0 & 0 \\
					\hline
					0 & 0 & 0 & 1 & 0 & 1 & 0 & 0 \\
					0 & 1 & 0 & 1 & 0 & 1 & 0 & 0 \\
					1 & 0 & 0 & 1 & 0 & 1 & 1 & 1 \\
					1 & 1 & 0 & 1 & 0 & 1 & 0 & 0 \\
					2 & 0 & 0 & 1 & 1 & 1 & 1 & 1 \\
					2 & 1 & 0 & 1 & 1 & 1 & 0 & 0 \\
					\hline
					0 & 0 & 1 & 0 & 2 & 0 & 0 & 0 \\
					0 & 1 & 1 & 0 & 2 & 0 & 0 & 0 \\
					1 & 0 & 1 & 0 & 2 & 0 & 1 & 1 \\
					1 & 1 & 1 & 0 & 2 & 0 & 0 & 0 \\
					2 & 0 & 1 & 0 & 2 & 0 & 1 & 1 \\
					2 & 1 & 1 & 0 & 2 & 0 & 0 & 0 \\
					\hline
					0 & 0 & 1 & 1 & 2 & 1 & 0 & 0 \\
					0 & 1 & 1 & 1 & 2 & 1 & 0 & 0 \\
					1 & 0 & 1 & 1 & 2 & 1 & 1 & 1 \\
					1 & 1 & 1 & 1 & 2 & 1 & 0 & 0 \\
					2 & 0 & 1 & 1 & 2 & 1 & 1 & 1 \\
					2 & 1 & 1 & 1 & 2 & 1 & 0 & 0 \\
					\hline
				\end{tabular}
			}
		\end{minipage}
		\quad
		\begin{minipage}{.325\textwidth}
			\centerline{
				\scalebox{.75}{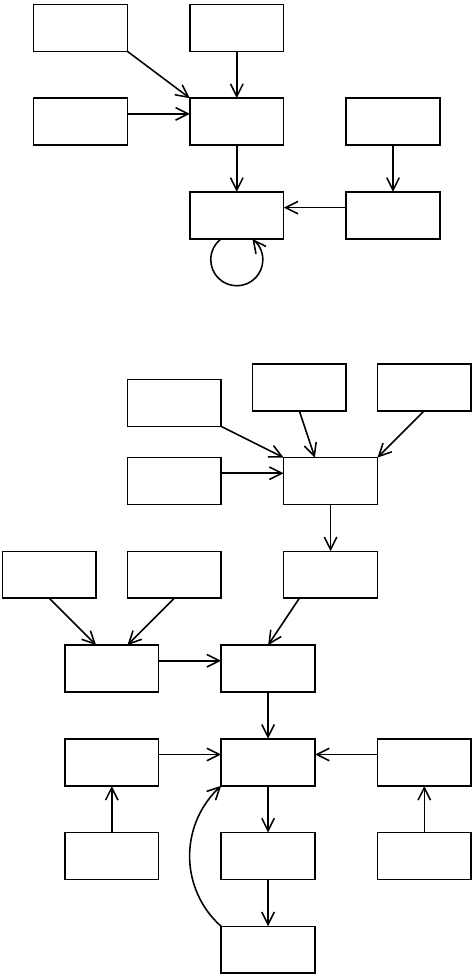}
			}
		\end{minipage}\medskip
		
		\begin{minipage}{.175\textwidth}
			\centerline{\footnotesize (a)}
		\end{minipage}
		\quad
		\begin{minipage}{.4\textwidth}
			\centerline{\footnotesize (b)}
		\end{minipage}
		\quad
		\begin{minipage}{.325\textwidth}
			\centerline{\footnotesize (c)}
		\end{minipage}
	\end{center}
	\caption{(a) Interaction graph, (b) transition tables of its local transition 
		functions and (c) transition graph of the GPBN of Example~\ref{example1}, 
		with $dt = (2, 1)$.}
	\label{fig:GPBN}
\end{figure}

First, let us analyze the equivalence between GPBNs and MBNs that have similar 
characteristics (as delay memory). Now, if we eliminate the intermediate translation 
of the genetic state $1$ to the protein state from the GPBN framework, both behaviors 
are the same. More precisely, if in the GPBN model we define $f_i : \B^N \to \{0,1\}$ 
as
\begin{equation*}
	\forall i,\ x_i(t+1) = f_i(x(t)) = \begin{cases}
		1 & \text{if } \Delta_i(t + 1) \geq 1\\
		f_{G_i}(x(t)) & \text{if } \Delta_i(t + 1) = 0
	\end{cases}\text{,}
\end{equation*}
and we take $x = (x_1, \dots, x_N)$ where, $\forall i \in \{1, \dots, N\}$, $x_i = 
x_{P_i}$ and
\begin{equation*}
	\begin{array}{rl}
		\Delta_i(0) & = \begin{cases}
			0 & \text{if } x_{i}(0) = 0\\
			\alpha \in \{1, \ldots, dt_i\} & \text{if } x_{i}(0) = 1
		\end{cases}\\
		\Delta_i(t+1) & = \begin{cases}
			0 & \text{if } f_{G_i}(x(t)) = 0 \text{ and } \Delta_i(t) = 0\\
			\Delta_i(t)-1 & \text{if } f_{G_i}(x(t)) = 0 \text{ and } 
				\Delta_i(t) > 0\\
			dt_i & \text{if } f_{G_i}(x(t)) = 1\\
		\end{cases}\text{,}
	\end{array}
\end{equation*}
we get the MBN model. 

\begin{example}\normalfont
	\label{example1}
	Consider the GPBN defined by means of the local transition functions $f_{G_1} = 
	f_{G_2} = x_{P_1} \land \neg x_{P_2}$ and delay vector $dt = (2, 1)$. The 
	interaction graph of this network in Figure~\ref{fig:GPBN}.a. Denoting each state 
	$x_{P_i} = k$ when $x_{P_i} = 1$ with variable memory $k$, \ie if delays $dt = 
	(dt_1, dt_2) = (2, 1)$ then $(x_{P_1}, x_{P_2}) \in \{0,1,2\} \times \{0,1\}$, 
	then we transform the model into this equivalent form in which delays are 
	integrated to the protein states:
	\begin{equation*}
		x_{P_1}(t+1) = \begin{cases}
			\Delta_{1}(t+1) & \text{if } \Delta_{1}(t+1) \geq 1\\
			x_{G_1}(t) & \text{if } \Delta_{1}(t+1) = 0
		\end{cases}
		\quad \text{and}\quad
		x_{P_2}(t+1) =x_{G_2}(t)\text{.}
	\end{equation*}
	Figures~\ref{fig:GPBN}.b and c picture the dynamical behavior of the GPBN, 
	emphasizing the existence of two attractors, the fixed point $(0,0,0,0)$ and a 
	limit cycle of size $3$. Now, by eliminating  the intermediate translation from 
	gene state $1$ to the protein state, this GPBN can be easily transformed into a 
	MBN whose vertex set is $V = \{P_1, P_2\}$, local transition functions are:
	\begin{equation*}
		x_{P_i}(t+1) = f_i(x(t)) = \begin{cases}
			1 & \text{if } \Delta_i(t+1) \geq 1\\
			x_{P_1}(t) \land \neg x_{P_2}(t) & \text{if } \Delta_i(t+1) = 0
		\end{cases}\text{,}
	\end{equation*}
	and delay vector is $dt=(2,1)$, as pictured in Figure~\ref{fig:MBN}.
	\begin{figure}[t!]
	\begin{center}
		\begin{minipage}{.36\textwidth}
			\centerline{\scalebox{.85}{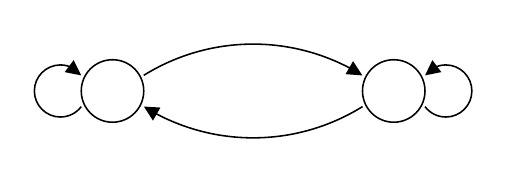}}
		\end{minipage}
		\quad
		\begin{minipage}{.275\textwidth}
			\centerline{\setlength{\tabcolsep}{2pt}\begin{tabular}{|cc|cc|}
				\hline
				$x_{P_1}$ & $x_{P_2}$ & $f_{1} $ & $f_{2}$\\
				\hline\hline
				0 & 0 & 0 & 0\\
				0 & 1 & 0 & 0\\
				1 & 0 & 2 & 1\\
				1 & 1 & 0 & 0\\
				2 & 0 & 2 & 1\\
				2 & 1 & 1 & 0\\
				\hline
			\end{tabular}}
		\end{minipage}
		\quad
		\begin{minipage}{.235\textwidth}
			\scalebox{.85}{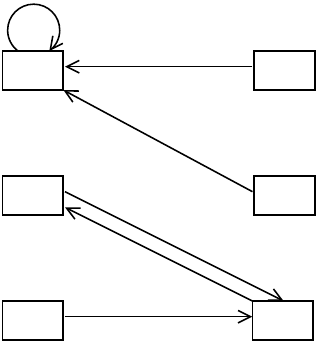}
		\end{minipage}\smallskip
	
		\begin{minipage}{.36\textwidth}
			\centerline{\scalebox{.85}{(a)}}
		\end{minipage}
		\quad
		\begin{minipage}{.275\textwidth}
			\centerline{\scalebox{.85}{(b)}}
		\end{minipage}
		\quad		
		\begin{minipage}{.235\textwidth}
			\centerline{\scalebox{.85}{(c)}}
		\end{minipage}\medskip
	\end{center}
	\caption{(a) Interaction graph, (b) transition tables of its local transition 
		functions $f_1$ and $f_2$ and (c) transition graph of the MBN related to 
		the GPBN given in Example~\ref{example1}.}
	\label{fig:MBN}
    \end{figure}
\end{example}

Now, let us see to what extent the MBN model is equivalent to the BN model. To do so, 
let us consider a MBN $F$ defined over an interaction graph $G = (V = \{1, \dots, 
n\}, E)$ with $n$ Boolean functions $(f_1, \dots, f_n)$ and delay vector $dt$. The 
idea is to find a BN $\tilde{F}$ without delays that simulates the asymptotic 
dynamics of $F$. To do so, to every vertex $v \in V$ we associate a set of vertices 
$\{[v,1], \dots, [v,dt_v]\}$, consider the neighborhood of vertex $v$, $\mathcal{N}_v 
= \{j \in V\ |\ (j, v) \in E\} = \{j_1, \dots, j_r\}$. So in the new network 
(constructed by replacing each vertex $v$ by a set of $dt_v$ vertices) we consider 
the following Boolean functions:
\begin{equation*}
	\begin{array}{llcl}
		\forall v \in V,\ &
		\tilde{f}_{[v, dt_v]}
		(x_{[1,1]}, \dots, x_{[n,dt_n]}) & 	
		= &
		f_v(x_{[j_1, 1]}, \dots, x_{[j_r, 1]})\\
		&
		\tilde{f}_{[v, dt_v-1]}
		(x_{[1,1]}, \dots, x_{[n,dt_n]}) &
		= & 
		f_v(x_{[j_1, 1]}, \dots, x_{[j_r, 1]}) \lor x_{[v, dt_v]}\\
		&& 
		\vdots &
		\\
		&
		\tilde{f}_{[v,1]}
		(x_{[1, 1]}, \dots, x_{[n,dt_n]}) & 	
		= & 
		f_v(x_{[j_1, 1]}, \dots, x_{[j_r, 1]}) \lor x_{[v,2]}
	\end{array}\text{,}
\end{equation*}
where the evolution of $\tilde{F}$ is such that $\tilde{x}_{[v, dt_v]}(t + 1) = 
f_v(x_{[j_1, 1]}(t), \dots, x_{[j_r, 1]}(t))$ and $\tilde{x}_{[v, a]}(t + 1) = 
f_v(x_{[j_1, 1]}(t), \dots, x_{[j_r, 1]}(t)) \lor x_{[v, a+1]}(t)$, $1 \leq a < 
dt_v$. This construction emphasizes an injective encoding $\phi$ of MBNs into BNs 
such that for any configuration $x$ of $F$, if $x \mapsto F(x)$ then $\tilde{x} = 
\phi(x) \mapsto \tilde{F}(\phi(x)) = \tilde{F}(\tilde{x})$, where $\tilde{x}$ is a 
configuration of $\tilde{F}$. Thus, if there exists a fixed point (resp. a limit 
cycle) among the attractors of $F$, $\tilde{F}$ admits also a fixed point (resp. a 
limit cycle) that is the encoding of the latter according to the given construction.
\begin{figure}[t!]
	\centerline{\scalebox{.85}{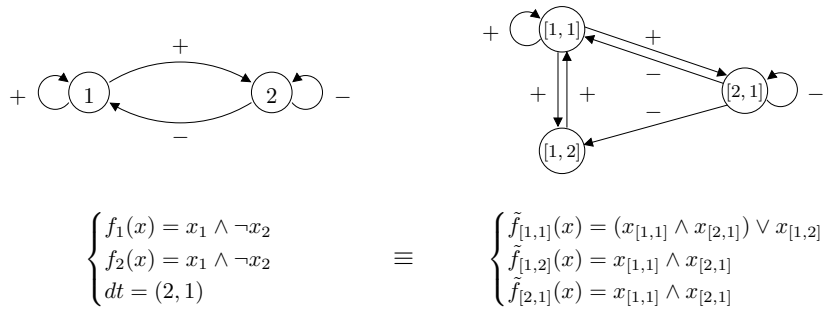}}
	\caption{(left) Interaction graph and local transition functions defining the 
		MBN of Example~\ref{example1} and (right) its associated and equivalent BN.}
	\label{fig:BN}
\end{figure}

\begin{example}\normalfont
	\label{example2}\normalfont
	Let us consider the MBN of Example~\ref{example1}, in which we rename each 
	protein vertex $P_i$ by its index $i$ such that $V = \{1, 2\}$. From the 
	construction given above, we obtain its related BN as pictured at the right of 
	Figure~\ref{fig:BN} whose local transition function truth tables and transition 
	graph are illustrated in Figure~\ref{fig:BN_dyn}.
	
	Notice that by construction of our simulation, $x_{[1,1]}(t) = 0$ implies that 
	$x_{[1,2]}(t) = 0$. More precisely, $x_{[1,1]}(t) = 0 \iff x_{[1,1]}(t) =  
	f_1(x_{[1, 1]}(t-1), x_{[2, 1]}(t-1)) \lor x_{[1,2]}(t-1) = 0$ that implies  
	that $x_{[1,2]}(t-1) = 0$ and $f_1(x_{[1, 1]}(t-1), x_{[2, 1]}(t-1))=0$. As a 
	consequence, since $x_{[1, 2]}(t) = f_1(x_{[1, 1]}(t-1), x_{[2, 1]}(t-1)) = 
	x_{[1, 1]}(t) = 0$, the configurations $x = (x_{[1,1]}, x_{[1, 2]}, 
	x_{[2, 1]}) = (0, 1, a)$, with $a \in \{0,1\}$, are artefacts of the 
	construction and are not real parts of the dynamical behavior.
	\begin{figure}[t!]
		\begin{center}
			\begin{minipage}{.38\textwidth}
				\centerline{\footnotesize \setlength{\tabcolsep}{1pt}
					\begin{tabular}{|ccc|ccc|}
						\hline
						$x_{[1,1]}$ & $x_{[1,2]}$ & $x_{[2,1]}$ &$f_{[1,1]}$ & 
						$f_{[1,2]}$ & $f_{[2,1]}$\\
						\hline
						\hline
						0 & 0 & 0 & 0 & 0 & 0\\
						0 & 0 & 1 & 0 & 0 & 0\\
						0 & 1 & 0 & 1 & 0 & 0\\
						0 & 1 & 1 & 1 & 0 & 0\\
						1 & 0 & 0 & 1 & 1 & 1\\
						1 & 0 & 1 & 0 & 0 & 0\\
						1 & 1 & 0 & 1 & 1 & 1\\
						1 & 1 & 1 & 1 & 0 & 0\\
						\hline
					\end{tabular}
				}
			\end{minipage}
			\qquad
			\begin{minipage}{.55\textwidth}
				\centerline{\scalebox{.85}{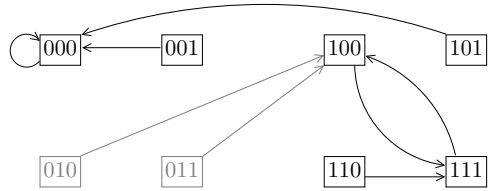}}
			\end{minipage}\medskip
		
			\begin{minipage}{.38\textwidth}
				\centerline{\footnotesize (a)}
			\end{minipage}\qquad
			\begin{minipage}{.55\textwidth}
				\centerline{\footnotesize (b)}
			\end{minipage}
		\end{center}
		\caption{(a) Truth tables of the local transition functions  $f_{[1,1]}$, 
			$f_{[2,1]}$, $f_{[1,2]}$ and (b) transition graph of the BN constructed in 
			Figure~\ref{fig:BN} that is equivalent to the MBN of 
			Example~\ref{example1}.}
		\label{fig:BN_dyn}
	\end{figure}
\end{example}

\subsection{Some particular classes of MBNs}

In this section, we present general results that hold for specific classes of MBNs. 
For our purpose, let us first focus on the class of positive disjunctive MBNs, \ie 
MBNs in which every vertex is associated with a local transition function composed 
only of the Boolean operator \textsc{or} with no negated variables. Before presenting 
any result, let us give two definitions~\cite{Brualdi1991}. First, the \emph{index of 
imprimitivity} $\eta(G)$ of a strongly connected digraph $G$ is the greatest common 
divisor of the lengths of all cycles of $G$. Second, the adjacency matrix $M$ of a 
strongly connected digraph $G$ is \emph{primitive} if and only if $M$ is an 
irreducible square matrix for which there exists a positive integer $m$ such that 
$\forall k \geq m$, $M^k$ is a strictly positive matrix. These definitions are 
related by the fact that, given a strongly connected digraph, its index of 
imprimitivity equals $1$ if and only if its adjacency matrix is primitive. 
Proposition~\ref{prop:disj_MBN} below emphasizes that such networks cannot admit 
limit cycles. 

\begin{proposition}
	\label{prop:disj_MBN}
	Let $F$ be a strongly connected positive disjunctive MBN and let $G = (V,E)$ be 
	its interaction graph. If $F$ is such that at least one vertex $v$ admits a delay 
	greater than $1$ ($dt_v \geq 2$), then $F$ does not have any limit cycle and can 
	only converge towards two fixed points: $\fpzero$ and $(dt_1, \dots, dt_n)$.
\end{proposition}

\begin{proof} 
	Let us consider $F$ such that it is composed of a vertex $v_0 \in V$ of delay 
	$dt_{v_0} \geq 2$. Following the construction proposed above, we can obtain a BN 
	$\tilde{F}$ of interaction graph $\tilde{G} = (\tilde{V}, \tilde{E})$ that 
	simulates $F$. $G$ being strongly connected by hypothesis, $v_0$ 
	belongs to a cycle $C$ in $G$. Let us admit that $C$ is of length $q$ and such 
	that:\smallskip
	
	\centerline{$C =$ \quad \scalebox{.85}{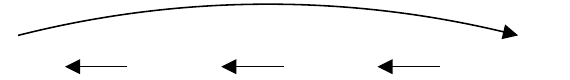}.}\smallskip
	
	\noindent From this, we derive that, in $\tilde{G}$, there are at least the two
	following cycles $\tilde{C}$ and $\tilde{C}'$ (the latter being the direct 
	consequence of $dt_{v_0} \geq 2$):\smallskip
	
	\centerline{$\tilde{C} =$ \quad \scalebox{.85}{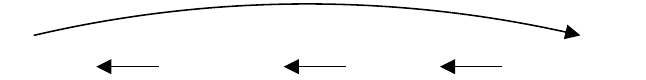}}\smallskip
	
	\noindent and:
	
	\centerline{$\tilde{C}' =$ \quad 
		\scalebox{.85}{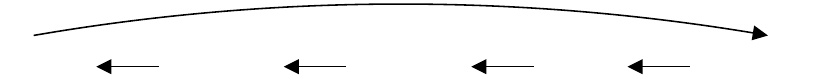},}\smallskip

	\noindent of respective lengths $|\tilde{C}|$ and $|\tilde{C}'|$ such that 
	$|\tilde{C}'| = |\tilde{C}| + 1$, which implies that 
	$\gcd(|\tilde{C}|,|\tilde{C}'|) = 1$. Thus, $\eta(\tilde{G}) = 1$ and $\tilde{F}$ 
	is primitive. In this context, this means that for every configuration $x \neq 
	\fpzero$, there exists $t \in \N$ such that $x(t) = x \cdot \tilde{M}^t = 
	\fpone$, where $\tilde{M}$ is the adjacency matrix of 
	$\tilde{G}$~\cite{Goles1980}. As a result, $\tilde{F}$ does not admit any 
	limit cycle but two fixed points $\fpzero$ and $\fpone$. From this, and because 
	of the disjunctive nature of the underlying MBN, the only attractors of $F$ are 
	$\fpzero$ and $(dt_1, \dots, dt_n)$.
\end{proof}

\begin{example}\normalfont
	\label{example3}
	In order to illustrate Proposition~\ref{prop:disj_MBN}, let us consider the 
	positive disjunctive BN defined as $(f_1(x) = x_2, f_2(x) = x_1 \lor x_3, 
	f_3(x) = x_2)$. This BN admits three attractors: the fixed points $\fpzero$ and 
	$\fpone$ and the limit cycle $010 \leftrightarrows 101$. Consider now the related 
	MBN with the same local transition functions in which $dt = (2,1,1)$. For this 
	network, the only attractors are $\fpzero$ and $(2,1,1)$. 
\end{example}

Because positive disjunctive BNs may admit fixed points and limit cycles (cf. Example~\ref{example3}), 
Proposition~\ref{prop:disj_MBN} highlights that \emph{the introduction of ``memory'' 
may freeze the dynamics} (\ie limit cycles may disappear). Now, the question is to 
know if this freezing property is invariant under the addition of delays. 
Proposition~\ref{prop:dag} below shows that it is true for digraphs with no cycles 
except possibly positive loops. Notice that this result is an extension of that of 
Robert~\cite{Robert1986} about acyclic BNs.

\begin{proposition}
	\label{prop:dag}
	Let $F$ be a BN whose interaction graph $G = (V, E)$  such that $|V| = n$ does 
	not induce cycles except possibly positive loops. Every MBN built on $F$ admits 
	only fixed points.
\end{proposition}

\begin{proof}
	With or without positive loops, $G$ can be represented by layers of different 
	depths as a classical directed acyclic graph. Now, let us focus on the first 
	layer $\mathcal{L}_1 = \{i\ |\ f_i(x) \text{ is constant } \lor f_i(x) = x_i\}$ , 
	\ie the layer that contains only vertices that are either source vertices or 
	mimic vertices. By definition of the local transition functions, all the vertices 
	of this layer will remain fixed after at most $\max_{i\in \mathcal{L}_1}(dt_i)$ 
	time steps. Once fixed, an induction on the layers depths suffices to show that 
	the stability of the whole MBN is reached and, considering that there are $k$ 
	layers in the graph, that it is reached in at most $\sum_{j=1}^k\max_{i \in 
	\mathcal{L}_j}{dt_i}$ time steps.
\end{proof}

The BN class of the previous proposition is not the only one for which the freezing 
property remains invariant. As Proposition~\ref{prop:decreasing} states, it is also 
the case for decreasing (resp. increasing) BNs that are such that $\forall x \in 
\B^n, F(x) \leq x$ (resp. $F(x) \geq x$), where $F(x) \leq x$ 
(resp. $F(x) \geq x$) if and only if $\forall i \in \{1, \dots, n\}, f_i(x) \leq x_i$ (resp. $f_i(x) \geq x_i$).

\begin{proposition}
	\label{prop:decreasing}
	Let $F$ be a decreasing (resp. increasing) BN. Any MBN associated to $F$ 
	necessarily converges towards fixed points only.
\end{proposition}

\begin{proof}
	Let us consider the case of a decreasing BN $F$ where $\forall x \in \B^n, F(x) 
	\leq x$. Such a global transition function implies locally that if 
	$x_i = 0$, it remains at that state because no transitions from state $0$ to 
	state $1$ are possible. As a consequence, whatever the delay vector $dt$ we 
	decide to associate to $F$ to create a MBN, the network cannot admit limit 
	cycles. The case of an increasing BN is analogous.
\end{proof}

However, despite the two classes previously presented, \emph{there obviously exist 
MBNs for which the introduction of ``memory'' does not lead to freeze their 
dynamics.} For instance, Figure~\ref{fig_BN_example} depicts a BN of size $2$ that 
admits only one fixed point, configuration $\fpzero$. Figure~\ref{fig:MBN} pictures 
the related MBN with $dt = (2,1)$ in which a limit cycle of length $2$ appears. 
Another more general example follows. Let us consider the BN $F: \B^n \to \B^n$ 
defined as:
\begin{equation*}
	F(x) = \begin{cases}
		(0,\dots,0, \stackbin{a\ b}{0{,}1}, 0,\dots,0) & \text{if } x = (0,\dots,0, 
			\stackbin{a\ b}{1{,}0}, 0,\dots,0)\\
		(0,\dots,0, \stackbin{a\ b}{1{,}0}, 0,\dots,0) & \text{if } x = (0,\dots,0, 
			\stackbin{a\ b}{1{,}1}, 0,\dots,0)\\
		\fpzero & \text{otherwise} 
	\end{cases}\text{.}
\end{equation*}
Clearly, this BN converges towards the fixed point $\fpzero$. However, adding delay 
can make a limit cycle appear. For instance, consider the associated MBN such that 
$dt = (\fpone \stackbin{a}{2} \fpone)$, \ie $\forall i \neq a, dt_i = 1$ and $dt_a = 
2$. It is easy to see that this MBN admits a limit cycle of length $2$ that is:
$(0,\dots,0, \stackbin{a\ b}{1{,}1}, 0,\dots,0) \leftrightarrows (0,\dots,0, \stackbin{a\ b}{2{,}0}, 0,\dots,0)$.

\section{MBN with two genes}
\label{sec:MBN2}

In this section, we analyze the dynamical behavior of every network composed of two 
vertices that admit fixed points. The idea is to highlight some of the main 
theoretical features of such interaction networks, which constitutes a first step for 
further formal studies of more general MBNs. 

\subsection{Networks admitting one fixed point}
\label{sec:MBN2_1fp}

\setlength{\tabcolsep}{2pt}
\begin{table}[t!]
	{\scriptsize \centerline{
		$\begin{array}{|c||m{.072\textwidth}|m{.072\textwidth}|
			m{.072\textwidth}|m{.072\textwidth}|m{.072\textwidth}|m{.072\textwidth}|
			m{.072\textwidth}|m{.072\textwidth}|m{.072\textwidth}|}
			\hline
			x & \makebox[.072\textwidth][c]{[1{,}00]} & 
				\makebox[.072\textwidth][c]{[2{,}00]} & 
				\makebox[.072\textwidth][c]{[3{,}00]} & 
				\makebox[.072\textwidth][c]{[4{,}00]} & 
				\makebox[.072\textwidth][c]{[5{,}00]} & 
				\makebox[.072\textwidth][c]{[6{,}00]} & 
				\makebox[.072\textwidth][c]{[7{,}00]} & 
				\makebox[.072\textwidth][c]{[8{,}00]} & 
				\makebox[.072\textwidth][c]{[9{,}00]}\\
			\hline\hline
			00 & \makebox[.072\textwidth][c]{00} & \makebox[.072\textwidth][c]{00} & 
				\makebox[.072\textwidth][c]{00} & \makebox[.072\textwidth][c]{00} & 
				\makebox[.072\textwidth][c]{00} & \makebox[.072\textwidth][c]{00} & 
				\makebox[.072\textwidth][c]{00} & \makebox[.072\textwidth][c]{00} & 
				\makebox[.072\textwidth][c]{00}\\
			01 & \makebox[.072\textwidth][c]{00} & \makebox[.072\textwidth][c]{00} & 
				\makebox[.072\textwidth][c]{00} & \makebox[.072\textwidth][c]{00} & 
				\makebox[.072\textwidth][c]{00} & \makebox[.072\textwidth][c]{00} & 
				\makebox[.072\textwidth][c]{00} & \makebox[.072\textwidth][c]{00} & 
				\makebox[.072\textwidth][c]{00}\\
			10 & \makebox[.072\textwidth][c]{00} & \makebox[.072\textwidth][c]{00} & 
				\makebox[.072\textwidth][c]{00} & \makebox[.072\textwidth][c]{01} & 
				\makebox[.072\textwidth][c]{01} & \makebox[.072\textwidth][c]{01} & 
				\makebox[.072\textwidth][c]{11} & \makebox[.072\textwidth][c]{11} & 
				\makebox[.072\textwidth][c]{11}\\
			11 & \makebox[.072\textwidth][c]{00} & \makebox[.072\textwidth][c]{01} & 
				\makebox[.072\textwidth][c]{10} & \makebox[.072\textwidth][c]{00} & 
				\makebox[.072\textwidth][c]{01} & \makebox[.072\textwidth][c]{10} & 
				\makebox[.072\textwidth][c]{00} & \makebox[.072\textwidth][c]{01} & 
				\makebox[.072\textwidth][c]{10}\\
			\hline
		\end{array}$
	}}\smallskip
	
	{\scriptsize \centerline{$\begin{array}{|c||m{.072\textwidth}|m{.072\textwidth}|
		m{.072\textwidth}|m{.072\textwidth}|m{.072\textwidth}|m{.072\textwidth}|
		m{.072\textwidth}|m{.072\textwidth}|m{.072\textwidth}|}
		\hline
		x & \makebox[.072\textwidth][c]{[10{,}00]} & 
			\makebox[.072\textwidth][c]{[11{,}00]} & 
			\makebox[.072\textwidth][c]{[12{,}00]} & 
			\makebox[.072\textwidth][c]{[13{,}00]} & 
			\makebox[.072\textwidth][c]{[14{,}00]} & 
			\makebox[.072\textwidth][c]{[15{,}00]} & 
			\makebox[.072\textwidth][c]{[16{,}00]} & 
			\makebox[.072\textwidth][c]{[17{,}00]} & 
			\makebox[.072\textwidth][c]{[18{,}00]}\\
		\hline\hline
		00 & \makebox[.072\textwidth][c]{00} & \makebox[.072\textwidth][c]{00} & 
			\makebox[.072\textwidth][c]{00} & \makebox[.072\textwidth][c]{00} & 
			\makebox[.072\textwidth][c]{00} & \makebox[.072\textwidth][c]{00} & 
			\makebox[.072\textwidth][c]{00} & \makebox[.072\textwidth][c]{00} & 
			\makebox[.072\textwidth][c]{00}\\
		01 & \makebox[.072\textwidth][c]{10} & \makebox[.072\textwidth][c]{10} & 
			\makebox[.072\textwidth][c]{10} & \makebox[.072\textwidth][c]{10} & 
			\makebox[.072\textwidth][c]{10} & \makebox[.072\textwidth][c]{10} & 
			\makebox[.072\textwidth][c]{10} & \makebox[.072\textwidth][c]{10} & 
			\makebox[.072\textwidth][c]{10}\\
		10 & \makebox[.072\textwidth][c]{00} & \makebox[.072\textwidth][c]{00} & 
			\makebox[.072\textwidth][c]{00} & \makebox[.072\textwidth][c]{01} & 
			\makebox[.072\textwidth][c]{01} & \makebox[.072\textwidth][c]{01} & 
			\makebox[.072\textwidth][c]{11} & \makebox[.072\textwidth][c]{11} & 
			\makebox[.072\textwidth][c]{11}\\
		11 & \makebox[.072\textwidth][c]{00} & \makebox[.072\textwidth][c]{01} & 
			\makebox[.072\textwidth][c]{10} & \makebox[.072\textwidth][c]{00} & 
			\makebox[.072\textwidth][c]{01} & \makebox[.072\textwidth][c]{10} & 
			\makebox[.072\textwidth][c]{00} & \makebox[.072\textwidth][c]{01} & 
			\makebox[.072\textwidth][c]{10}\\
		\hline
	\end{array}$}}\smallskip
	
	{\scriptsize \centerline{$\begin{array}{|c||m{.072\textwidth}|m{.072\textwidth}|
		m{.072\textwidth}|m{.072\textwidth}|m{.072\textwidth}|m{.072\textwidth}|
		m{.072\textwidth}|m{.072\textwidth}|m{.072\textwidth}|}
		\hline
		x & \makebox[.072\textwidth][c]{[19{,}00]} & 
			\makebox[.072\textwidth][c]{[20{,}00]} & 
			\makebox[.072\textwidth][c]{[21{,}00]} & 
			\makebox[.072\textwidth][c]{[22{,}00]} & 
			\makebox[.072\textwidth][c]{[23{,}00]} & 
			\makebox[.072\textwidth][c]{[24{,}00]} & 
			\makebox[.072\textwidth][c]{[25{,}00]} & 
			\makebox[.072\textwidth][c]{[26{,}00]} & 
			\makebox[.072\textwidth][c]{[27{,}00]}\\
		\hline\hline
		00 & \makebox[.072\textwidth][c]{00} & \makebox[.072\textwidth][c]{00} & 
			\makebox[.072\textwidth][c]{00} & \makebox[.072\textwidth][c]{00} & 
			\makebox[.072\textwidth][c]{00} & \makebox[.072\textwidth][c]{00} & 
			\makebox[.072\textwidth][c]{00} & \makebox[.072\textwidth][c]{00} & 
			\makebox[.072\textwidth][c]{00}\\
		01 & \makebox[.072\textwidth][c]{11} & \makebox[.072\textwidth][c]{11} & 
			\makebox[.072\textwidth][c]{11} & \makebox[.072\textwidth][c]{11} & 
			\makebox[.072\textwidth][c]{11} & \makebox[.072\textwidth][c]{11} & 
			\makebox[.072\textwidth][c]{11} & \makebox[.072\textwidth][c]{11} & 
			\makebox[.072\textwidth][c]{11}\\
		10 & \makebox[.072\textwidth][c]{00} & \makebox[.072\textwidth][c]{00} & 
			\makebox[.072\textwidth][c]{00} & \makebox[.072\textwidth][c]{01} & 
			\makebox[.072\textwidth][c]{01} & \makebox[.072\textwidth][c]{01} & 
			\makebox[.072\textwidth][c]{11} & \makebox[.072\textwidth][c]{11} & 
			\makebox[.072\textwidth][c]{11}\\
		11 & \makebox[.072\textwidth][c]{00} & \makebox[.072\textwidth][c]{01} & 
			\makebox[.072\textwidth][c]{10} & \makebox[.072\textwidth][c]{00} & 
			\makebox[.072\textwidth][c]{01} & \makebox[.072\textwidth][c]{10} & 
			\makebox[.072\textwidth][c]{00} & \makebox[.072\textwidth][c]{01} & 
			\makebox[.072\textwidth][c]{10}\\
		\hline
	\end{array}$}}
	\caption{Exhaustive list of the $27$ BNs that admits configuration $(0,0)$ as 
		their unique fixed point.}
	\label{tab:BN2_fp00}
\end{table}

\begin{figure}
	\centerline{\scalebox{.55}{}}
\end{figure}

Here, we focus on networks that admit a unique fixed point. Because the other cases 
can be studied similarly (and do not reveal other dynamical peculiarities), we pay 
only attention to networks that converge towards the fixed point $(0,0)$. First of 
all, it is trivial to list all the $27$ BNs that admit $(0,0)$ as their unique fixed 
point. For the sake of clarity, let us introduce the following notation that allows 
to specify the $27$ distinct networks at stake here.
\begin{notation}
	\label{nota:MBN2_1fpnets}
	$[k{,}x]$, with $k \in \{1, \dots, 27\}$ and $x \in \B^2$, denotes the network of 
	size $2$ whose local transition functions are represented by $k$ and defined by 
	their truth tables in Table~\ref{tab:BN2_fp00} and that admits $x$ as its unique 
	fixed point (represented as a binary word).
\end{notation}
Let us also denote by $\F[x]$ the set composed of all BNs that admit $x$ as their 
unique fixed point, $\FP[x] \subseteq \F[x]$ the set of BNs for which $x$ is the 
unique attractor, and $\LC[x] \subseteq \F[x]$ the set of BNs that admit at least a 
limit cycle.


\renewcommand{\tabcolsep}{5pt}
\begin{figure}[t!]
	\centerline{\scalebox{.47}{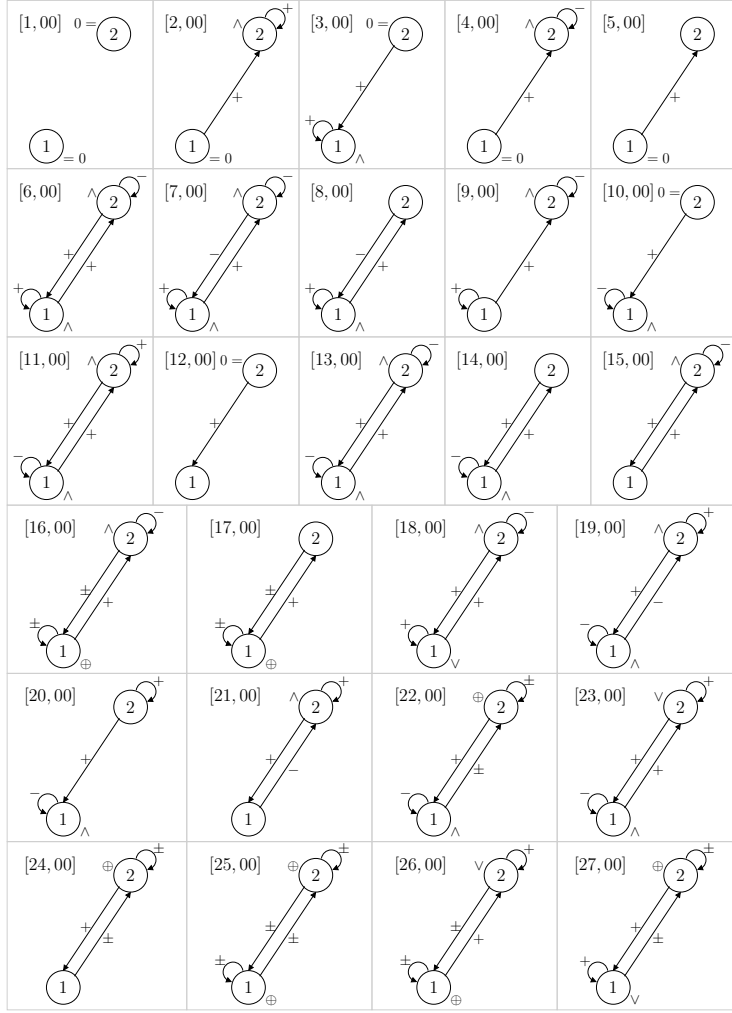}}
	\caption{Interaction graphs of all the networks that admit $(0,0)$ as a fixed 
		point, whose local transition functions are given in Table~\ref{tab:BN2_fp00}.
		In some of these graphs, abusing notations for not burdening the reading, an 
		arc $(x,y)$ that is labeled by $\pm$ stands for the actual existence of two 
		arcs from $x$ to $y$, one labeled by $+$, the other by $-$. Such an arc 
		highlights notably that the local transition function of $y$ is not locally 
		monotonic, and more precisely a \textsc{xor} (denoted by the operator 
		$\oplus$) in this case.}
	\label{fig:BN2_fp00}
\end{figure}

From now on, let us analyze the dynamical behavior of networks belonging to $\F[x]$. 
Notice that in general, besides the fixed point $x$, a BN may (or may not) admit 
limit cycles. The peculiar question that we address now deals with the asymptotic 
behavior invariance when we add delays to a BN and thus change it into a MBN. 
Actually, what follows aims at determining the delay vectors $dt = (dt_1, dt_2)$ for 
which a given network $f \in F[x]$ admits a particular attractor (say the fixed point 
$x$, or a limit cycle which appears only when delays are added). First of all, it is 
easy to see from Table~\ref{tab:BN2_fp00} that 
\begin{multline*}
	\FP[00]\ =\ \{
		[1,00], [2,00], [3,00], [4,00], [5,00], [6,00], [7,00], [8,00],\\ 
			[10,00], [11,00], [12,00], [16,00], [19,00], [21,00], [22,00], [25,00]
	\}\text{,}
\end{multline*}
and that
\begin{multline*}
	\LC[00]\ =\ \{
		[9,00], [13,00], [14,00],[15,00], [17,00], [18,00],\\
		[20,00], [23,00], [24,00], [26,00], [27,00]
	\}.
\end{multline*}
\begin{notation}
	\label{nota:MBN2_00}
	In what follows, we will make particular use of the following notations:
	\begin{itemize}
	\item Delay vector $dt = (dt_1, dt_2)$ will be denoted by 
		$dt = (\alpha, \beta)$.
	\item Initial configuration $(\Delta_1, \Delta_2)$ will be denoted by 
		$(\rho, \gamma)$, with $0 \leq \rho \leq \alpha$ and $0 \leq \gamma \leq 
		\beta$.
	\item We will write $(x_1(0), x_2(0)) \xrightarrow{n} (x_1(n), x_2(n))$ to refer 
		to $(x_1(0), x_2(0)) \to (x_1(1), x_2(1)) \to \cdots \to (x_1(n), 
		x_2(n))$. Furthermore, we will use $x \xrightarrow{*} x'$ to indicate that 
		$x'$ belongs to the configurations that are successors of $x$.
	\end{itemize}
\end{notation}

Let us now study separately both classes $\FP[00]$ and $\LC[00]$.

\subsubsection{Analysis of $\FP[00]$}

Let us divide $\FP[00]$ into two sub-classes $\FP_\stab^{[00]}$ and $\FP_\cy^{[00]}$ 
such that:
\begin{itemize}
\item $\FP_\stab^{[00]} = \{[1,00], \dots, [5,00], [8,00], [10,00], [12,00], 
[21,00]\}$, and 
\item $\FP_\cy^{[00]} = \FP[00] \setminus \FP_\stab^{[00]}$. 
\end{itemize}
Propositions~\ref{prop:FP00fp} and~\ref{prop:FP00lc} below respectively show that 
MBNs related to BNs of $\FP_\stab^{[00]}$ cannot admit limit cycles whereas MBNs 
related to BNs of $\FP_\cy^{[00]}$ can.

\begin{proposition}
	\label{prop:FP00fp}
	For any delay vector $dt$, all the MBNs built on a BN belonging to 
	$\FP_\stab^{[00]}$ have a unique attractor, the fixed point $(0,0)$. 
\end{proposition}

\begin{proof}	
	Let us first consider BN $[1,00]$. Because its local transition functions 
	are both constant and equal to $0$, it is trivial that any MBN built on it can 
	only converge towards fixed point $(0,0)$, in at most $\max(\alpha,\beta)$ time 
	steps.\smallskip

	Now, let us consider any MBN built on the BNs that belongs to the subset of 
	$\FP_\stab^{[00]}$ defined as $\{[2,00], [3,00], [5,00], [12,00]\}$ such that $dt 
	= (\alpha, \beta)$. As pictured in Figure~\ref{fig:BN2_fp00}, their interaction 
	graphs do not induce cycles of length greater than or equal to $2$ and the only 
	loops they contain are positive. So, from Proposition~\ref{prop:dag}, we derive 
	that $(0,0)$ is the unique attractor of MBNs built on them and that it is reached 
	in at most $\alpha+\beta$ time steps.\smallskip
	
	About MBNs built on BN $[4,00]$ with any initial configuration $(\rho, \gamma)$, 
	because of function $f_1(x)$ that is constant and equal to $0$, $x_1$ is 
	necessarily fixed to $0$ after $\rho \leq \alpha$ time steps. Once $x_1 = 0$, by 
	definition of $f_2(x)$, $x_2$ will decrease to reach $0$ after at most $\beta$ 
	time steps. A similar reasoning can be used on the basis of BN $[10,00]$ because 
	they are symmetric networks. Moreover, they both converge towards $(0,0)$ in at 
	most $\alpha + \beta$ time steps.\smallskip
	
	Now, concerning MBNs built on BN $[8,00]$, let us consider two cases. The first 
	case is when $x_1 = 0$. By definition of $f_1(x)$, once $x_1$ at $0$ it remains 
	fixed to $0$ and according to the definition of $f_2$, the value of $x_2$ 
	necessarily decreases to $0$. Thus, $(0,x_2)$ converges towards $(0,0)$ in at 
	most $\beta$ time steps. Consider now the case where $0 < x_1 \leq \alpha$. We 
	have: if $x_2 \geq 1$ then $(x_1, x_2) \to (x_1 - 1, \beta) \xrightarrow{x_1-1} 
	(0, \beta) \xrightarrow{\beta} (0, 0)$, and the network converges in 
	$x_1+\beta$ time steps; if $x_2 = 0$ then $(x_1, 0) \to (\alpha, \beta) \to 
	(\alpha-1, \beta) \xrightarrow{\alpha-1} (0, \beta) \xrightarrow{\beta} (0,0)$ 
	and the network converges in $\alpha+\beta+1$ time steps. Thus, such MBNs (resp. 
	MBNs built on BN $[21,00]$ by symmetry) converge towards their unique attractor, 
	fixed point $(0,0)$, in at most $\alpha+\beta+1$ time steps.\smallskip
		
	Hence, all the MBNs of $\FP_\stab^{[00]}$ admit $(0,0)$ as their unique 
	attractor.
\end{proof}

\begin{proposition}
	\label{prop:FP00lc}
	For any BN of $\FP_\cy^{[00]}$, there exist delay vectors $dt$s such that their 
	related MBNs admit limit cycles. 
\end{proposition}

\begin{proof}
	In this proof, for BNs belonging to $\FP_\cy^{[00]}$, we exhibit delay 
	vectors with which the associated MBN evolves towards a limit cycle.
	
		Let us now consider a MBN built on BN $[6,00]$. First, suppose that $\alpha = 1$. 
	Then, for all $\beta$, either $\rho = 0$ and $x_1$ will stay fixed to $0$, which 
	leads inevitably $x_2$ to decrease to $0$ and remain fixed in $\gamma$ time 
	steps, or $\rho = 1$: in this case: if $\gamma = 0$, then after one time step, 
	$x_1 = 0$ and we get back to the previous case; otherwise, $\gamma \geq 1$ and 
	the dynamics is $(1,\gamma) \xrightarrow{\gamma} (1,0) \to (0, \gamma)$, and we 
	get back to previous item in $\gamma+1$ time steps. As a consequence, in order 
	for MBNs associated with BN $[6,00]$ to admit limit cycles, $\alpha$ needs to be 
	greater than $1$ (except if $\beta$ is also equal to $1$ of course). Now let us 
	consider $dt = (\alpha \geq 2, \beta)$ and initial configuration $(\rho = \alpha, 
	\gamma = \beta)$. Its dynamics is $(\alpha,\beta) \to (\alpha,\beta-1) 
	\xrightarrow{\beta-1} (\alpha, 0) \to (\alpha-1, \beta) \to (\alpha, \beta-1)$. 
	So, there exist recurrent configurations and thus a limit cycle. Notice that this 
	limit cycle is of length $\beta$. Furthermore, for the same reasons as those 
	given for the case where $\alpha = 1$, initial configurations such that $\rho = 
	0$ or $(\rho=1, \gamma=0)$ converges towards fixed point $(0,0)$ in at most 
	$\beta+1$ time steps. All the other configurations, \ie those such that $(\rho 
	\geq 2, \gamma \geq 0)$ and $(\rho = 1, \gamma = 0)$ evolves to $(\alpha, \beta)$ 
	and thus towards the limit cycle. By symmetry, a similar reasoning can be used 
	for BN $[11,00]$.\smallskip
	
	About MBNs built on BN $[7,00]$, for all $(\alpha, \beta)$, notice that by 
	definition of $f_1$ and because of the positive loop, when $x_1 = 0$, it remains 
	fixed to $0$. Moreover, when $x_1 = 0$, it leads $x_2$ to decrease until it 
	reaches $0$ also. Let us now detail the dynamics of this network depending on the 
	initial configurations, considering $1 \leq \rho \leq \alpha$ and $1 \leq \gamma 
	\leq \beta$:
	\begin{itemize}
	\item if $\rho \leq \gamma$ then $(\rho, \gamma) \xrightarrow{\rho} (0, 
		\gamma - \rho) \xrightarrow{\gamma-\rho} (0,0)$ that is reached in $\gamma$ 
		time steps.
	\item if $\rho > \gamma$ then $(\rho, \gamma) \xrightarrow{\gamma} (\rho - 
		\gamma, 0) \to (\alpha, \beta)$. From configuration $(\alpha,\beta)$, we have:
		if $\alpha = \beta$ then $(\alpha, \alpha) \to (\alpha-1, \alpha-1) 
		\xrightarrow{\alpha-1} (0,0)$ that is reached in $\gamma+\alpha+1$ time steps; 
		if $\alpha < \beta$ then $(\alpha, \beta) \xrightarrow{\alpha} (0, 
		\beta-\alpha) \xrightarrow{\beta-\alpha} (0,0)$ that is reached in 
		$\gamma+\beta+1$ time steps; if $\alpha > \beta$ then $(\alpha, \beta) 
		\xrightarrow{\beta} (\alpha-\beta, 0) \to (\alpha,\beta)$, which emphasizes a 
		limit cycle of length $\beta$.
	\end{itemize}
	Thus, any MBN built on BN $[7,00]$ such that $\alpha > \beta$ admits a limit 
	cycle of length $\beta$. The unique other possible attractor is fixed point 
	$(0,0)$ that is reached in at most $2\max(\alpha,\beta)+1$ time steps. By 
	symmetry, a similar reasoning can be used for BN $[19,00]$.\smallskip

	Now, let us focus on BN $[16,00]$.
	First of all, by definition of its local transition functions, it is easy to 
	check that for all $(\alpha, \beta)$, if $\rho = \gamma$ then the configuration 
	converges towards fixed point $(0,0)$. Indeed, we have $(\rho, \rho) \to (\rho-1, 
	\rho-1) \xrightarrow{\rho-1} (0,0)$, that is reached in $\rho$ time steps, \ie in 
	at most $\min(\alpha,\beta)$ time steps. Now, let us deal with all possible 
	$\alpha$ and $\beta$. First, consider that $\alpha = \beta$. Then, according to 
	the initial configuration, if $\rho < \gamma$, the dynamics of $(\rho, \gamma)$ 
	is $(\rho, \gamma) \xrightarrow{\rho} (0, \gamma - \rho) \to (\alpha, 
	\gamma - \rho - 1) \xrightarrow{\gamma - \rho - 1 < \alpha} (\alpha - \gamma + 
	\rho + 1, 0) \to (\alpha, \alpha) \xrightarrow{\alpha} (0,0)$ that is reached in 
	$\alpha+\gamma+1$ time steps; otherwise, if $\rho > \gamma$, it is $(\rho, 
	\gamma) \xrightarrow{\gamma} (\rho - \gamma, 0) \to (\alpha, \alpha) 
	\xrightarrow{\alpha} (0,0)$ that is reached in $\alpha+\gamma+1$ time steps.
	Thus, if $\alpha = \beta$, any configuration converges towards fixed point 
	$(0,0)$ in at most $2\alpha+1 = 2\beta+1$ time steps. 
	
	Now, consider that $\alpha > \beta$. In this case, if $\rho < \gamma$, the 
	dynamics of $(\rho, \gamma)$ is $(\rho, \gamma) \xrightarrow{\rho} (0, \gamma - 
	\rho) \to (\alpha, \gamma - \rho - 1) \xrightarrow{\gamma - \rho - 1 < \alpha}
	(\alpha - \gamma + \rho + 1, 0) \to (\alpha, \beta) \xrightarrow{\beta} 
	(\alpha-\beta, 0) \to (\alpha, \beta)$; otherwise, if $\rho > \gamma$, the 
	dynamics of $(\rho, \gamma)$ is $(\rho, \gamma) \xrightarrow{\gamma} (\rho - 
	\gamma, 0) \to (\alpha, \beta) \xrightarrow{\beta} (\alpha - \beta, 0) \to 
	(\alpha, \beta)$. Thus, when $\alpha > \beta$, for all initial configurations 
	$(\rho, \gamma)$, with $\rho \neq \gamma$, the network evolves towards a limit 
	cycle of length $\beta$. 
	
	Now, in the case where $\alpha < \beta$, we have:
	\begin{itemize}
	\item if $\rho > \gamma$, the dynamics is $(\rho, \gamma) \xrightarrow{\gamma} 
		(\rho - \gamma, 0) \to (\alpha, \beta)$. Here, suppose that $\alpha < \beta < 
		2\alpha+1$. Then we have $(\alpha, \beta) \xrightarrow{\alpha} (0, \beta - 
		\alpha) \to (\alpha, \beta - \alpha - 1) \xrightarrow{\beta - \alpha - 1} 
		(\alpha - (\beta - \alpha - 1), 0) \to (\alpha, \beta)$. Now, if $\beta = 
		2\alpha + 1$, the trajectory of $(\alpha, \beta)$ is $(\alpha, \beta) 
		\xrightarrow{\alpha} (0, \beta - \alpha) \to (\alpha, \beta - \alpha - 1) 
		\xrightarrow{\alpha} (0, \beta - 2\alpha -1) = (0, 0)$, that is reached in at 
		most $2\beta+1$ time steps. More generally, if there exists $k \in \N$ such 
		that $\beta = k \cdot (\alpha + 1) - 1$ then, the trajectory of $(\alpha, 
		\beta)$ is $(\alpha, \beta) \xrightarrow{k \cdot (\alpha + 1) - 1} (0, \beta - 
		k \cdot (\alpha + 1) - 1) = (0,0)$, that is reached from any $(\rho, \gamma)$ 
		in at most $2\beta+1$ time steps. Now, suppose that $\forall k \in \N, \beta 
		\neq k \cdot (\alpha + 1) - 1$ and $\ell$ the greatest natural number such 
		that $\beta > \ell \cdot (\alpha + 1) - 1$ (\ie $\ell \cdot (\alpha + 1) - 1 < 
		\beta < (\ell + 1) (\alpha + 1 - 1)$), then the trajectory of $(\alpha, 
		\beta)$ is $(\alpha, \beta) \xrightarrow{\ell \cdot (\alpha + 1) - 1} (0, 
		\beta - (\ell \cdot (\alpha + 1) - 1) \neq 0) \to (\alpha, \beta - \ell \cdot 
		(\alpha + 1)) \xrightarrow{\beta - \ell \cdot (\alpha + 1)} (\alpha - (\beta - 
		\ell \cdot (\alpha + 1)) \neq 0, 0) \to (\alpha, \beta)$.
	\item if $\rho < \gamma$, if we suppose that $\gamma = \rho + \alpha + 1$, the 
		dynamics is $(\rho, \gamma) \xrightarrow{\rho} (0, \gamma - \rho) \to (\alpha, 
		\gamma - \rho - 1) \xrightarrow{\gamma - \rho - 1} (\alpha - \gamma + \rho + 1 
		= 0, 0)$, that is reached in at most $\beta$ time steps. More generally, by 
		supposing that given $k > 1 \in \N$, $\gamma = \rho + k \cdot (\alpha + 1)$, 
		we have $(\rho, \gamma) \xrightarrow{\rho} (0, \gamma - \rho) \to (0, \gamma - 
		\rho - 1) \xrightarrow{\alpha}\\ (0, \gamma - \rho - 1 - \alpha) 
		\xrightarrow{(k - 1)(\alpha + 1)} (0, \gamma - \rho - k\cdot(\alpha  + 1)) = 
		(0,0)$, that is reached in at most $\beta$ time steps. Now, suppose that 
		$\forall k \in \N, \gamma \neq \rho + k \cdot (\alpha + 1)$ and that $\ell$ is 
		the greatest natural number such that $\gamma > \rho + \ell \cdot (\alpha + 
		1)$, \ie $\rho + \ell \cdot (\alpha + 1) < \gamma < \rho + (\ell + 1) \cdot 
		(\alpha + 1)$, the trajectory of $(\rho, \gamma)$ is $(\rho, \gamma) 
		\xrightarrow{\rho + \ell \cdot (\alpha + 1)} (0, \gamma - \rho - \ell \cdot 
		(\alpha + 1)) \to (\alpha, \gamma - \rho - \ell \cdot (\alpha + 1) - 1) 
		\xrightarrow{\gamma - \rho - \ell \cdot (\alpha + 1) - 1} (\alpha - (\gamma - 
		\rho - \ell \cdot (\alpha + 1) - 1) \neq 0, 0) \to (\alpha, \beta)$, for which 
		it suffices to apply the case discussed in the previous sub-item. 
	\end{itemize}
	As a consequence, when $\alpha < \beta$, initial configurations can evolve 
	either towards the fixed point $(0,0)$ in at most $2\beta+1$ time steps or 
	towards a limit cycle of length $\beta$ only if $\beta \neq k \cdot (\alpha + 1) 
	- 1, \forall k \in \N$. By symmetry, a similar reasoning can be used for BN 
	$[22,00]$.

	Concerning MBNs built on BN $[25,00]$, first of all, because this network is a 
	symmetric \textsc{xor} network, it is easy to check that if $\rho = \gamma$, 
	whatever $\alpha$ and $\gamma$ are, the trajectory of the initial configuration 
	is $(\rho, \rho) \to (\rho - 1, \rho - 1) \xrightarrow{\rho-1} (0,0)$, and	
	converges thus towards fixed point $(0,0)$ in $\rho$ time steps, \ie in at most 
	$\min(\alpha,\beta)$ time steps. From now on, let us focus on initial 
	configurations such that $\rho \neq \gamma$. Suppose that $\alpha > 
	\beta$. The underlying dynamics is $(\rho, \gamma) \xrightarrow{\min(\rho, 
	\gamma)+1} (\alpha, \beta) \xrightarrow{\beta} (\alpha - \beta, 0) \to (\alpha, 
	\beta)$, which highlights a limit cycle of length $\beta$. Symmetrically, in the 
	case where $\alpha < \beta$, $(\rho, \gamma)$ evolves towards the cycle $(\alpha, 
	\beta) \xrightarrow{\alpha} (0, \beta - \alpha) \to (\alpha, \beta)$. Lastly, if 
	$\alpha = \beta$, the trajectory of $(\rho, \gamma)$ is $(\rho, \gamma) 
	\xrightarrow{\min(\rho,\gamma) + 1} (\alpha, \alpha) \xrightarrow{\alpha} 
	(0,0)$, and there is convergence towards the unique fixed point in at most 
	$2\alpha+1$ time steps.
\end{proof}

Proposition~\ref{prop:FP00lc} emphasizes that \emph{adding delays to BNs to make them 
become MBNs can lead to the creation of limit cycles in the set of network 
attractors}. Now, let us focus on the class $\LC[00]$.

\subsubsection{Analysis of $\LC[00]$}

By definition, the BNs that belong to $\LC[00]$ have the feature of having, 
besides the fixed point $(0,0)$, a limit cycle. This limit cycle is of length $2$ for 
$[\{9, 13, 14, 15, 18, 20, 23, 26, 27\}, 00]$ and of length $3$ for $[\{17, 24\}, 
00]$. Often, in BN models of real genetic regulatory networks, the biological meaning 
comes from the fixed point. Indeed, except in networks modeling biological rhythms 
sustained oscillations in which limit cycles are of course meaningful, the latter 
correspond to spurious asymptotic behaviors. Thus, in a modeling framework, finding a 
way of removing limit cycles can be particularly relevant. Classically, it is done by 
using an asynchronous updating mode. As Example~\ref{example3} highlighted it, adding 
delays can serve in this context. More generally, we will see that changing BNs into 
MBNs may allow to obtain this desirable property of avoiding spurious attractors. 
More precisely, for all the BNs that belong to $\LC[00]$, we analyze the delay 
parameter space for knowing the regions in which the limit cycle disappears. This 
analysis is presented by the following propositions. 

\begin{proposition}
	\label{prop:LC9}
	Every MBN built on BN $[9,00]$ (resp. on $[20,00]$) admits a limit cycle that is 
	reached by all configurations such that $\rho \geq 1$ (resp. $\gamma \geq 1$). 
\end{proposition}

\begin{proof}
	Consider a MBN built on BN $[9,00]$. Whatever $\alpha$ and $\beta$ are, by 
	definition of $f_1$, $x_1$ is maintained by the positive loop. As a consequence: 
	if $\rho = 0$, the network converges towards the fixed point $(0,0)$ in $\gamma$ 
	time steps, \ie in at most $\beta$ time steps; if $\rho \geq 1$, by definition of 
	$f_2$, the initial configuration $(\rho, \gamma = \beta)$ has the following 
	dynamics: $(\rho, \beta) \to (\rho, \beta-1) \xrightarrow{\beta-1} (\rho, 0) \to 
	(\rho, \beta)$. Thus, any configuration $(\rho \geq 1, \gamma)$ belongs to a 
	limit cycle of length $\beta$. By symmetry, a similar reasoning can be used for 
	BN $[20,00]$.
\end{proof}

\begin{proposition}
	\label{prop:LC13}
	Consider a MBN $M$ based on BN $[13,00]$. We have:
	\begin{enumerate}
	\item If $M$ is such that $\gcd(\alpha+1, \beta+1) = 1$, all its configurations 
		converge towards fixed point $(0,0)$.
	\item If $M$ is such that $\gcd(\alpha+1, \beta+1) > 1$, any configuration 
		$(\rho,\gamma)$ such that $\rho + \ell_0(\alpha+1) = \gamma + k_0(\beta+1)$, 
		with $0 \leq k_0 \leq k$ and $0 \leq \ell_0 \leq \ell$ with $\ell(\alpha+1) = 
		k(\beta+1) = \lcm(\alpha+1, \beta+1)$, converges towards fixed point 
		$(0,0)$. Otherwise, it evolves towards a limit cycle of length $\lcm(\alpha+1, 
		\beta+1)$.
	\end{enumerate}
\end{proposition}

\begin{proof}
	In this proof, we deal with the two distinct items of the statement separately.
	\begin{enumerate}
	\item Let us first consider a MBN such that $\gcd(\alpha+1, \beta+1) = 1$ and 
		such that $\alpha < \beta$. Let $(\rho_t, \gamma_t)$ be the configuration 
		obtained after $t \in \N$ time steps from any possible initial configuration 
		$(\rho, \gamma)$. In this case, inevitably, there exists $t$, with $0 \leq t <
		(\alpha+1) \cdot \beta + 1$, such that $\rho_t = \gamma_t$. Indeed, suppose 
		on the contrary that $\forall t \in \N$ such that $0 \leq t < (\alpha+1) \cdot 
		\beta + 1,\ \rho_t \neq \gamma_t$. Then, we know that there exist $(\alpha+1) 
		\cdot \beta$ different ordered pairs of integers (\ie configurations) such 
		that $\rho_t \neq \gamma_t$. Thus, necessarily, there exist $i \neq j \in \{0, 
		\dots, (\alpha+1) \cdot \beta \}$ such that $(\rho_i, \gamma_i) = (\rho_j, 
		\gamma_j)$, which implies the existence of a limit cycle of length $|i-j|$. 
		More precisely, given $0 \leq h \leq \beta$, $k < \alpha + 1$ and $\ell < 
		\beta + 1$, we have:
		\begin{equation*}
			(\rho_i, \gamma_i) \xrightarrow{|i-j|} (\rho_i, \gamma_i)\ \implies\ 
			(\alpha, \beta - h) \xrightarrow{|i-j| = \ell(\alpha+1) = k(\beta+1)} 
			(\alpha, \beta - h)\text{,}
		\end{equation*}
		which implies that $\alpha + 1$ and $\beta + 1$ are not coprime, which is a 
		contradiction. Therefore, there exists a time step $t$ at which $\rho_t = 
		\gamma_t$. Now, given the local transition functions, it is easy to remark 
		that all the configurations whose two terms are equal converge towards fixed 
		point $(0,0)$.
		
		By symmetry of the rule, the same reasoning applies for the case where $\alpha 
		> \beta$. Moreover, by the hypothesis highlighting that $\gcd(\alpha+1, 
		\beta+1) = 1$, the case where $\alpha = \beta$ does not exist. 

	\item Now, let us consider a MBN such that $\gcd(\alpha+1, \beta+1) = \lambda > 
		1$. Consider an initial configuration $(\rho, \gamma)$ and distinguish two 
		cases:
		\begin{itemize}
		\item $(\rho, \gamma)$ satisfies 
			\begin{equation}
				\label{eq1_rule13}
				\rho + \ell_0(\alpha+1) = \gamma + k_0(\beta+1)
			\end{equation} 
			for some $k_0$, $\ell_0$ such that $0 \leq k_0 \leq k$, $0 \leq \ell_0 
			\leq \ell$, with $\ell(\alpha+1) = k(\beta+1) = \lcm(\alpha+1, \beta+1)$.
			First of all, as evoked above in the previous item, given the nature of 
			the MBN, it is trivial to remark that for all $dt = (\alpha, \beta)$, an 
			initial configuration such that $\rho = \gamma$ (this initial condition 
			satisfies Equation~\ref{eq1_rule13} above with $k_0 = \ell_0 = 0$) admits 
			the trajectory $(\rho, \rho) \to (\rho - 1, \rho - 1) \xrightarrow{(\rho - 
			1)} (0,0)$, and converges towards fixed point $(0,0)$ in $\rho$ time 
			steps, \ie in at most $\min(\alpha,\beta)$ time steps.
			Now, let us admit that $\alpha < \beta$ and that $\rho > \gamma$. If the 
			lower values of $k_0$ and $\ell_0$ satisfying Equation~\ref{eq1_rule13} 
			are both equal to $1$ then the configuration admits the following 
			trajectory: $(\rho, \gamma) \xrightarrow{\gamma} (\rho - \gamma, 0) \to 
			(\rho - \gamma - 1, \beta) \xrightarrow{\rho - \gamma - 1} (0, \beta - 
			(\rho - \gamma - 1)) \to (\alpha, \beta - (\rho - \gamma) = 
			\alpha)$. Moreover, if the lower values of $k_0$ and $\ell_0$ satisfying 
			Equation~\ref{eq1_rule13} are such that $k_0 > 1$ and $\ell_0 \geq 1$ then  
			the configuration admits the following trajectory: $(\rho, \gamma) 
			\xrightarrow{\rho + 1 + (\ell_0-1)(\alpha+1) = \gamma + 1 + k_0(\beta+1) - 
			(\alpha+1)} (\alpha,\alpha)$. 
			Hence, in both cases, $(\rho, \gamma)$ converges towards fixed point 
			$(0,0)$. 
			
			With the same reasoning, it can be shown that the result holds also for 
			initial configurations such that $\rho < \gamma$. Furthermore, by symmetry 
			of the rule, the same reasoning applies for the case where $\alpha > 
			\beta$. Moreover, if $\alpha = \beta$, the only way for 
			Equation~\ref{eq1_rule13} to hold is when $\rho = \gamma$ and this case 
			has already been dealt with.
			
		\item $(\rho, \gamma)$ satisfies $\rho + \ell_0(\alpha+1) \neq \gamma + 
			k_0(\beta+1)$ for some $k_0$, $\ell_0$ such that $0 \leq k_0 \leq k$, $0 
			\leq \ell_0 \leq \ell$, with $\ell(\alpha+1) = k(\beta+1) = \lcm(\alpha+1, 
			\beta+1)$. First, let us consider the case where $\alpha = \beta$. In this 
			case, the only way for the negation of Equation~\ref{eq1_rule13} to hold 
			is when $\rho \neq \gamma$. So, let us consider a configuration where 
			$\rho > \gamma$. Its trajectory is $(\rho, \gamma) \xrightarrow{\gamma} 
			(\rho - \gamma, 0) \to (\rho - \gamma - 1, \alpha) \xrightarrow{\rho - 
			\gamma - 1} (0, \alpha - (\rho - \gamma - 1)) \to (\alpha, \alpha - (\rho 
			- \gamma)) \xrightarrow{\alpha - (\rho - \gamma)} (\rho - \gamma, 
			0)$, which highlights a limit cycle of length $\alpha + 1$. With the same 
			reasoning, we can show that a configuration such that $\rho < \beta$ 
			evolves also towards a limit cycle of length $\alpha + 1$.
			Now, let us consider the case where $\alpha < \beta$. By the hypothesis 
			stating that $\rho + \ell_0(\alpha+1) \neq \gamma + k_0(\beta+1)$, the 
			trajectory of a configuration such that $\rho > \gamma$ is $(\rho, \gamma) 
			\xrightarrow{\gamma} (\rho - \gamma, 0) \to (\rho - \gamma - 1, \beta) 
			\xrightarrow{\rho - \gamma - 1} (0, \beta - (\rho - \gamma - 1)) \to 
			(\alpha, \beta - (\rho - \gamma)) \xrightarrow{\ell(\alpha+1) = 
			k(\beta+1)} (\alpha, \beta - (\rho - \gamma))$, which corresponds to a 
			limit cycle of length $\lcm(\alpha+1, \beta+1)$. In other terms, if 
			$\alpha+1$ and $\beta+1$ are not coprime and if $\rho + \ell_0(\alpha+1) 
			\neq \gamma + k_0(\beta+1)$ then the network admits a limit cycle.
			
			With the same reasoning, it can be shown that the result holds also for 
			initial configurations such that $\rho < \gamma$. Furthermore, by symmetry 
			of the rule, the same reasoning applies for the case where $\alpha > 
			\beta$.
		\end{itemize}
	\end{enumerate} 
\end{proof}

\begin{proposition}
	\label{prop:LC14-15_17-24_18-23_26-27}
	Let $S = \{[9,00],[13,00],[20,00]\}$. Every MBN built on a BN belonging to 
	$\LC[00] \setminus S$ admits a limit cycle that is reached by all 
	configurations except $(0,0)$.
\end{proposition}

\begin{proof} 
	Let us consider a MBN built on $[14,00]$ and an initial configuration such that 
	$\rho \geq 1$ and $\gamma \geq 2$. By definition of $f_2$, $x_2$ is set and stays 
	at $\beta$ while $x_1$ is positive. Notice also that $x_1$ decreases until 
	reaching $0$ by definition of $f_1$. Thus, we have the following trajectory: 
	$(\rho, \gamma) \to (\rho-1, \beta) \xrightarrow{\rho-1} (0, \beta) \to 
	(\alpha, \beta - 1) \to (\alpha-1, \beta) \xrightarrow{\alpha-1} (0, \beta)$,
	which highlights the evolution towards a limit cycle of length $\alpha+1$. Now, 
	let us focus on an initial configuration defined as $(\rho \geq 1, \gamma = 0)$. 
	Its trajectory is $(\rho, 0) \to (\rho-1, \beta) \xrightarrow{\rho - 1} (0, 
	\beta) \to (\alpha, \beta-1) \to (\alpha-1, \beta) \xrightarrow{\alpha - 1} (0, 
	\beta)$. As a consequence, for any MBN built on BN $[14,00]$ (resp. on BN 
	$[15,00]$ by symmetry), all configurations except fixed point $(0,0)$ evolves 
	towards a limit cycle of length $\alpha+1$ (resp. $\beta+1$). 
	
	Notice that a similar reasoning applies for MBNs based on BNs $[17,00]$, 
	$[23,00]$, $[26,00]$, and for their respective symmetric BNs $[24,00]$, $[18,00]$ 
	and $[27,00]$.
\end{proof}

Thanks to Propositions~\ref{prop:FP00fp} to~\ref{prop:LC14-15_17-24_18-23_26-27} 
above, we obtain the following theorem that recapitulates all the results that 
characterize the dynamical behaviors of all MBNs that can be constructed from BNs of 
size $2$ having the unique fixed point $(0,0)$.

\begin{theorem}
	\label{thm:FP00} 
	Table~\ref{tab:FP00} gives the dynamical behavior of any MBN built on the basis 
	of a BN belonging to $\FP(0,0)$, with delays $dt = (\alpha, \beta)$, initial 
	condition $0 \leq \rho \leq \alpha$, $0 \leq \gamma \leq \beta$ and $k,\ell \in 
	\N$.
\end{theorem}

\begin{table}
	\begin{small}	
		\centerline{\begin{tabular}{|c|c|c|c|}
			\hline
			Rules & Attractors & Delays $dt=(\alpha,\beta)$ & Initial conditions 
				$(\rho,\gamma)$\\
			\hline\hline
			\multirow{6}{1.9cm}{$[1,00],[2,00]{,}$ $[3,00],[4,00]{,}$ 
				$[5,00],[8,00]{,}$ $[10,00],[12,00]{,}$ $[21,00]$
			} & & &\\
			& {FP} & $\forall (\alpha, \beta)$ & $\forall (\rho, \gamma)$\\
			& &  & \\
			\cline{2-4} &  &  &\\ 
			& LC & $\emptyset$ & $\emptyset$\\
			&  &  &\\ 
			\hline\hline
			\multirow{3}{1.9cm}{$[6,00]$
			} & {FP} & $\alpha = 1$ & $\forall (\rho, \gamma)$\\
			& & $\forall (\alpha, \beta)$ & $(\rho=1 \land \gamma=0) \lor 
				(\rho=0 \land \gamma\geq0)$\\
			\cline{2-4} & LC & $(\alpha,\beta) \geq (2,1)$ & $(\rho \geq 2 \land 
				\gamma \geq 0) \lor (\rho=1 \land \gamma=1)$\\
			\hline
			\multirow{3}{1.9cm}{$[11,00]$
			} & 
			{FP} & $\beta = 1$ &  $\forall (\rho, \gamma)$\\
			& & $\forall (\alpha,\beta)$ &  ($\rho=0 \land \gamma=1) \lor (\rho \geq 0
				\land \gamma=0)$\\
			\cline{2-4} & LC & $(\alpha,\beta) \geq (1,2)$ & $(\rho \geq 0 \land 
				\gamma \geq 2) \lor (\rho=1 \land \gamma=1)$\\
			\hline\hline
			\multirow{2}{1.9cm}{$[7,00]$
			} & {FP} & $\alpha \leq \beta$ & $\forall (\rho, \gamma)$\\
			& & $\alpha > \beta$ & $\rho \leq \gamma$\\
			\cline{2-4}
			& LC &  $\alpha > \beta$ & $\rho > \gamma$\\
			\hline
			\multirow{2}{1.9cm}{$[19,00]$
			} & {FP} & $\alpha \geq \beta$ & $\forall (\rho, \gamma)$\\
			& & $\alpha < \beta$ & $\rho \geq \gamma$\\
			\cline{2-4}
			& LC &  $\alpha < \beta$ & $\rho < \gamma$\\
			\hline\hline
			\multirow{2}{1.9cm}{$[9,00]$
			} & {FP} & $\forall (\alpha,\beta)$ & $\rho = 0$\\
			\cline{2-4} & LC & $\forall (\alpha,\beta)$ & $\rho \geq 1$\\
			\hline
			\multirow{2}{1.9cm}{$[20,00]$
			} & {FP} & $\forall (\alpha,\beta)$ & $\gamma = 0$\\
			\cline{2-4} & LC & $\forall (\alpha,\beta)$ & $\gamma \geq 1$\\
			\hline\hline
			\multirow{4}{1.9cm}{$[13,00]$
			} & FP & $\gcd(\alpha+1,\beta+1) = 1$ & $\forall (\rho, \gamma)$\\
			& & $\gcd(\alpha+1,\beta+1) > 1$ & $\rho + \ell_0(\alpha+1) = \gamma + 
				k_0(\beta+1)$\\
			& & & $\text{with } 0 \leq k_0 \leq k, 0 \leq \ell_0 \leq \ell,$\\
			& & & $\ell(\alpha+1) = k(\beta+1) = \lcm(\alpha+1, \beta+1)$\\
			\cline{2-4} & LC & $\gcd(\alpha+1,\beta+1) > 1$ & $\rho + 
				\ell_0(\alpha+1) \neq \gamma + k_0(\beta+1)$\\
			& & & $\text{with } 0 \leq k_0 \leq k, 0 \leq \ell_0 \leq \ell,$\\
			& & & $\ell(\alpha+1) = k(\beta+1) = \lcm(\alpha+1, \beta+1)$\\			
			\hline\hline
			\multirow{4}{1.9cm}{$[14,00],[15,00]{,}$ $[17,00],[18,00]{,}$ 
				$[23,00],[24,00]{,}$ $[26,00]{,}[27,00]$
			} & {FP} & $\forall (\alpha, \beta)$ & $(\rho, \gamma) = (0,0)$\\
			& & & \\
			\cline{2-4} & {LC} & $\forall (\alpha, \beta)$ & $(\rho, \gamma) \neq 
				(0,0)$\\
			& & & \\
			\hline\hline
			\multirow{8}{1.9cm}{
				$[16,00]$
			} & {FP} & $\alpha < \beta$ & $(\rho = \gamma) \lor (\rho + k(\alpha+1) = 
				\gamma) \lor$\\
			& & & $(\beta + 1 = k(\alpha+1))$\\
			& & $\alpha = \beta$ & $\forall (\rho, \gamma)$\\
			& & $\alpha > \beta$ & $\rho = \gamma$\\ 
			\cline{2-4}
			& LC & $\alpha < \beta$ & $(\rho < \gamma) \land (\gamma \neq \rho + k 
				\cdot (\alpha + 1)) \land$\\
			& & & $(\beta + 1 \neq k(\alpha+1))$\\
			& & $\alpha < \beta$ & $(\rho > \gamma) \land (\beta + 1 
				\neq k(\alpha+1))$\\
			& & $\alpha > \beta$ & $\rho \neq \gamma$\\ 
			\hline		
			\multirow{8}{1.9cm}{
				$[22,(0,0)]$
			} & {FP} & $\alpha < \beta$ & $\rho = \gamma$\\
			& & $\alpha = \beta$ & $\forall (\rho, \gamma)$\\
			& & $\alpha > \beta$ & $(\rho = \gamma) \lor (\gamma + k(\beta + 1) = 
				\rho) \lor$\\
			& & & $(\alpha + 1 = k(\beta + 1))$\\
			\cline{2-4}
			& LC & $\alpha < \beta$ & $\rho \neq \gamma$\\
			& & $\alpha > \beta$ & $(\rho > \gamma) \land (\rho \neq \gamma + k 
				\cdot (\beta + 1)) \land$\\
			& & & $(\alpha + 1 \neq k(\beta + 1))$\\
			& & $\alpha > \beta$ & $(\rho < \gamma) \land (\alpha + 1 
				\neq k(\beta + 1))$\\
			\hline\hline
			\multirow{3}{1.9cm}{
				$[25,00]$
			} & {FP} & $\alpha = \beta $ & $\forall (\rho, \gamma)$\\
			& & $\alpha \neq \beta$ & $\rho = \gamma$\\
			\cline{2-4}
			& LC & $\alpha \neq \beta $ & $\rho \neq \gamma$  \\
			\hline
		\end{tabular}}
	\end{small}
	\caption{Dynamical behaviors of all the MBNs built on the basis of the BNs 
		belonging to $\FP(0,0)$, with delays $dt = (\alpha, \beta)$, initial condition 
		$0 \leq \rho \leq \alpha$, $0 \leq \gamma \leq \beta$ and $k, \ell \in \N$.}
	\label{tab:FP00}
\end{table}

Now the dynamical properties of MBNs of size $2$ with a unique fixed point have been 
characterized, let us pay attention to MBNs of the same size with two fixed points.

\subsection{Networks admitting two fixed points}
\label{sec:MBN2_2fp}

In this section, we focus on the MBNs that can be built on the basis of BNs of size 
$2$ that admit two fixed points. First of all, let us notice that there exist $6$ 
distinct classes of such networks, each of which being composed of $9$ networks. As 
what has been presented above, let us introduce the following notations. 
\begin{notation}
	$[k, x, y]$, with $k \in \{1, \dots, 9\}$ and $x, y \in \B^2$, denotes the 
	network of size $2$ whose local transition functions are represented by $k$ and 
	defined by their truth tables in 
	Tables~\ref{tab:FPs00-01},~\ref{tab:FPs00-11},~\ref{tab:FPs01-10} 
	and~\ref{tab:FPs01-11} and that admits $x$ and $y$ as its fixed points 
	(represented as binary words).
\end{notation}
Let us also denote by $\F[x, y]$ the set composed of all BNs that admit $x$ and $y$ 
as their unique fixed points, $\FP[x, y] \subseteq \F[x, y]$ the set of BNs for which 
$x$ and $y$ are the unique attractors, $\LC[x, y] \subseteq \F[x, y]$ the set of BNs 
that admits at least a limit cycle.

\subsubsection{MBNs based on $\F[00, 01]$ and $\F[00, 10]$}
\label{sec:MBN2_2fp_00_01}

Concerning the classes of BNs $\F[00, 01]$ and $\F[00, 10]$, notice first that they 
are symmetric. Thus, all the results obtained for $\F[00, 01]$ have their symmetric 
that hold for $\F[00, 10]$. So, let us focus only on $\F[00, 01]$.

\begin{table}[t!]
	{\scriptsize \centerline{
		$\begin{array}{|c||c|c|c|c|c|c|c|c|c|}
			\hline
			x & [1,00,01] & [2,00,01] & [3,00,01] & [4,00,01] & [5,00,01] & 
				[6,00,01] & [7,00,01] & [8,00,01] & [9,00,01]\\
			\hline\hline
			00 & 00 & 00 & 00 & 00 & 00 & 00 & 00 & 00 & 00\\
			01 & 01 & 01 & 01 & 01 & 01 & 01 & 01 & 01 & 01\\
			10 & 00 & 00 & 00 & 01 & 01 & 01 & 11 & 11 & 11\\
			11 & 00 & 01 & 10 & 00 & 01 & 10 & 00 & 01 & 10\\
			\hline
		\end{array}$}}\smallskip
		
	\hspace*{1pt}{\scriptsize \centerline{
		$\begin{array}{|c||c|c|c|c|c|c|c|c|c|}
			\hline
			x & [1,00,10] & [2,00,10] & [3,00,10] & [4,00,10] & [5,00,10] & 
				[6,00,10] & [7,00,10] & [8,00,10] & [9,00,10]\\
			\hline\hline
			00 & 00 & 00 & 00 & 00 & 00 & 00 & 00 & 00 & 00\\
			01 & 00 & 00 & 00 & 10 & 10 & 10 & 11 & 11 & 11\\
			10 & 10 & 10 & 10 & 10 & 10 & 10 & 10 & 10 & 10\\
			11 & 00 & 01 & 10 & 00 & 01 & 10 & 00 & 01 & 10\\
			\hline
		\end{array}$}}
	\caption{Truth tables of all the $9$ BNs that admit (up) fixed points $(0,0)$ and 
		$(0,1)$ and (down) fixed points $(0,0)$ and $(1,0)$.}
	\label{tab:FPs00-01}
\end{table}

\begin{remark}
	\label{rem:MBN2_2fp_00_01}
	For every network of $\F[00, 01]$, since $(0,0)$ and $(0,1)$ are fixed points, 
	whatever $\alpha$ and $\beta$ are, configuration $(0,0)$ cannot change, admits 
	the following trajectory $(0,0)$\loopr{} and is a fixed point, and the trajectory 
	of any initial configuration such that $\rho = 0$ and $\gamma \geq 1$ is 
	$(0, \gamma) \to (0, \beta)$\loopr{} that leads to fixed point $(0, \beta)$.
\end{remark}
Therefore we analyze the dynamical behavior of all the initial configurations of the 
form $(\rho, \gamma)$ where $\rho \neq 0$. From Table~\ref{tab:FPs00-01} (up), a 
basic enumeration gives that $\FP[00, 01] = \F[00,01] \setminus \{[9,00,01]\}$, and 
$\LC[00, 01] = \{[9,00,01]\}$. Now, let us partition $\FP[00, 01]$ into the following 
two sub-classes: $\FP_\stab^{[00,01]} = \{[1,00,01], \dots, [5,00,01], 
[8,00,01]\}$, and $\FP_\cy^{[00,01]} = \FP[00,01] \setminus \FP_\stab^{[00,01]} = 
\{[6,00,01],[7,00,01]\}$. Proposition~\ref{prop:MBN2_2fp_00_01_FPstab} below shows 
that MBNs based on BNs of $\FP_\stab[00,01]$ admits only two fixed points, $(0,0)$ 
and $(0,\beta)$.

\begin{proposition}
	\label{prop:MBN2_2fp_00_01_FPstab}
	For any delay vector $dt$, every MBN built on a BN that belongs to 
	$\FP_\stab^{[00,01]}$ admits only two attractors, fixed points $(0,0)$ and 
	$(0,\beta)$.
\end{proposition}

\begin{proof}
	Let us first consider MBNs built on BNs $[1,00,01]$, $[2,00,01]$ or $[3,00,01]$. 
	Table~\ref{tab:FPs00-01} (up) shows that they are decreasing networks. As a 
	consequence, by Proposition~\ref{prop:decreasing}, they admit only fixed points 
	that are $(0,0)$ and $(0,\beta)$ by definition of the local transition 
	functions. All the MBNs built on $[1,00,01]$, $[2,00,01]$ (resp. on $[3,00,01]$) 
	converge in at most $\alpha$ (resp. $\alpha+\beta$) time steps.\smallskip
	
	Moreover, Table~\ref{tab:FPs00-01} (up) highlights also that the interaction 
	graph of BN $[5,00,01]$ does not induce cycles except a positive loop on vertex 
	$2$. Thus, by Proposition~\ref{prop:dag}, any MBN built on $[5,00,01]$ admits 
	only fixed points that are $(0,0)$ and $(0,\beta)$. More precisely, for any delay 
	vector $dt = (\alpha, \beta)$ and $\rho \geq 1$, we have the following 
	trajectory: $(\rho, \gamma) \xrightarrow{\rho} (0,\beta)$\loopr{} that is reached 
	in at most $\alpha$ time steps.\smallskip
	
	Consider a MBN built on BN $[4,00,01]$. Let $(\rho, \gamma)$ be any initial 
	configuration. If $\rho = \gamma$, its trajectory is $(\rho,\rho) 
	\xrightarrow{\rho} (0,0)$\loopr{} and the network converges in at most 
	$\min(\alpha, \beta)$ time steps. If $\rho < \gamma$, its trajectory is $(\rho, 
	\gamma) \xrightarrow{\rho} (0, \gamma - \rho) \to (0, \beta)$\loopr{} and the 
	network converges in $\alpha+1$ at most. Now, if $\rho > \gamma$, the trajectory 
	begins by $(\rho,\gamma) \xrightarrow{\gamma} (\rho - \gamma, 0) \to (\rho - 
	\gamma - 1, \beta)$, and the trajectory of $(\rho - \gamma - 1, \beta)$ is either 	$(\rho - \gamma - 1, \beta) \xrightarrow{\rho - \gamma - 1} (0,\beta)$\loopr{} if 
	$\rho - \gamma \neq k(\beta+1)$, or $(\rho - \gamma - 1, \beta) \xrightarrow{\rho 
	- \gamma - 1} (0,0)$\loopr{} otherwise. Thus the network converges towards these 
	two fixed points in at most $\alpha$ time steps.\smallskip

	Consider finally a MBN built on BN $[8,00,01]$. For any delay vector $dt = 
	(\alpha, \beta)$, we have: if $\rho \geq 1$ and $\gamma \geq 1$ then $(\rho, 
	\gamma) \to (\rho-1, \beta) \xrightarrow{\rho-1} (0,\beta)$\loopr{} that is 
	reached in $\rho$ time steps, and if $\rho \geq 1$ and $\gamma = 0$ then $(\rho, 
	0) \to (\alpha, \beta) \xrightarrow{\alpha} (0, \beta)$\loopr{} that is reached 
	in $\alpha+1$ time steps. Thus, the network admits only the two fixed points 
	$(0,0)$ and $(0,\beta)$ and its convergence time is at most $\alpha+1$ time 
	steps.
\end{proof}

Proposition~\ref{prop:MBN2_2fp_00_01_FPcy} shows that there exist specific conditions 
under which MBNs built on BNs belonging to $\FP_\cy[00,01]$ evolve towards a limit 
cycle.

\begin{proposition}
	\label{prop:MBN2_2fp_00_01_FPcy}
	For all the BNs of $\FP_\cy[00,01]$, there exist delay vectors $dt$s such that 
	any associated MBN admits a limit cycle.
\end{proposition}

\begin{proof}
	Consider BN $[6,00,01]$. Let us consider two cases for $\alpha$ and begin with 
	$\alpha > 1$. The different possible evolutions are: if $\rho \geq 1, \gamma \geq 
	1$ then $(\rho,\gamma) \xrightarrow{\gamma} (\alpha, 0) \to (\alpha-1, \beta) \to 
	(\alpha, \beta-1) \xrightarrow{\beta-1} (\alpha, 0)$, which emphasizes a limit 
	cycle of length $\beta+1$; if $\rho > 1, \gamma = 0$ then $(\rho, 0) \to (\rho-1, 
	\beta) \to (\alpha, \beta-1) \xrightarrow{\beta-1} (\alpha,0)$ that belongs to a 
	limit cycle of length $\beta+1$; if $\rho = 1, \gamma = 0$ then $(1, 0) \to (0, 
	\beta)$\loopr{}. Now, consider that $\alpha = 1$. The different possible 
	evolutions are: if $\rho = 1, \gamma = 0$ then $(1,0) \to (0, \beta)$\loopr{}; if 
	$\rho = 1, \gamma \geq 1$ then $(1,\gamma) \xrightarrow{\gamma} (1, 0) \to (0, 
	\beta)$\loopr{} and reaches its fixed point in $\gamma+1$ time steps. Thus, the 
	MBNs built on BN $[6,00,01]$ converge to their fixed points in at most $\beta+1$ 
	time steps and can admit a limit cycle of length $\beta+1$.\medskip
	
	Now, consider BN $[7,00,01]$. Let us consider three cases depending on the 
	initial configuration $(\rho, \gamma)$. First, if $\rho > \gamma$, we have:
	if $\alpha > \beta$ then $(\rho, \gamma) \xrightarrow{\gamma} (\rho - \gamma, 0) 
	\to (\alpha, \beta) \xrightarrow{\beta} (\alpha-\beta, 0) \to (\alpha, \beta)$, 
	which emphasizes a limit cycle of length $\beta+1$; if $\alpha < \beta$ then 
	$(\rho, \gamma) \xrightarrow{\gamma} (\rho - \gamma, 0) \to (\alpha, \beta) 
	\xrightarrow{\alpha} (0, \beta-\alpha) \to (0, \beta)$\loopr{}, and reaches its 
	fixed point in $\alpha+\gamma+2$ time steps; if $\alpha = \beta$ then $(\rho,
	\gamma) \xrightarrow{\gamma} (\rho - \gamma, 0) \to (\alpha, \alpha) 
	\xrightarrow{\alpha} (0,0)$\loopr{}, and reaches its fixed point in 
	$\alpha+\gamma+1$ time steps. Now, whatever $dt$ is: if $\rho < \gamma$ then 
	$(\rho, \gamma) \xrightarrow{\rho} (0, \gamma - \rho) \to (0, \beta)$\loopr{}, 
	and reaches its fixed point in $\rho+1$ time steps; if $\rho = \gamma$ then 
	$(\rho, \rho) \xrightarrow{\rho} (0, 0)$\loopr{}, and reaches its fixed point in 
	$\rho$ time steps. Thus, the MBNs built on BN $[7,00,01]$ converge to their fixed 
	points in at most $\alpha + \beta + 2$ time steps and can admit a limit cycle of 
	length $\beta+1$.
\end{proof}

Proposition~\ref{prop:MBN2_2fp_00_01_FP9} shows the same principle for $[0,00,01] \in 
\LC[00,01]$.

\begin{proposition}
	\label{prop:MBN2_2fp_00_01_FP9}
	Considering $[9,00,01]$ the unique element of $\FP_\cy[00,01]$, for all delay 
	vectors $dt$, there exist initial conditions such that any associated MBN admits
	a limit cycle.
\end{proposition}

\begin{proof}
	Consider any $dt$ and an initial configuration such that $\rho \geq 1$. We have: 
	if $\gamma = 0$ then $(\rho, 0) \to (\alpha, \beta) \xrightarrow{\beta} (\alpha, 
	0) \to (\alpha, \beta)$, which emphasizes a limit cycle of length $\beta+1$; if 
	$\gamma \geq 1$ then $(\rho, \gamma) \xrightarrow{\gamma} (\alpha, \beta)$ that 
	belongs to a limit cycle of length $\beta+1$. Thus, the MBNs built on BN 
	$[9,00,01]$ converge to their fixed points in at most $1$ time step and can admit 
	a limit cycle of length $\beta+1$.
\end{proof}

\subsubsection{MBNs based on $\F[00, 11]$}
\label{sec:MBN2_2fp_00_11}

The class of BNs $\F[00, 11]$ does not admit a symmetric class and what follows gives 
a characterization of the dynamics of MBNs built on it. 

\begin{table}[t!]
	{\scriptsize \centerline{
		$\begin{array}{|c||c|c|c|c|c|c|c|c|c|}
			\hline
			x & [1,00,11] & [2,00,11] & [3,00,11] & [4,00,11] & [5,00,11] & 
				[6,00,11] & [7,00,11] & [8,00,11] & [9,00,11]\\
			\hline\hline
			00 & 00 & 00 & 00 & 00 & 00 & 00 & 00 & 00 & 00\\
			01 & 00 & 00 & 00 & 10 & 10 & 10 & 11 & 11 & 11\\
			10 & 00 & 01 & 11 & 00 & 01 & 11 & 00 & 01 & 11\\
			11 & 11 & 11 & 11 & 11 & 11 & 11 & 11 & 11 & 11\\
			\hline
		\end{array}$
	}}
	\caption{Truth tables of all the $9$ BNs that admit fixed points $(0,0)$ and 
		$(1,1)$.}
	\label{tab:FPs00-11}
\end{table}

\begin{remark}
	\label{rem:MBN2_2fp_00_11}
	For every network of $\F[00, 11]$, since $(0,0)$ and $(1,1)$ are fixed points, 
	whatever $\alpha$ and $\beta$ are, configuration $(0,0)$ cannot change, admits 
	the following trajectory $(0,0)$\loopr{} and is a fixed point, and any initial 
	configuration such that $\rho \geq 1$ and $\gamma \geq 1$ admits the following 
	trajectory $(\rho, \gamma) \to (\alpha, \beta)$\loopr{} that leads to fixed point 
	$(\alpha, \beta)$.
\end{remark}
Therefore we analyze the dynamical behavior of all the initial configurations of the 
form $(0, \gamma)$ where $\gamma \neq 0$ and $(\rho, 0)$ where $\rho \neq 0$. From 
Table~\ref{tab:FPs00-11}, a basic enumeration gives that $\FP[00, 11] = \F[00,11] 
\setminus \{[5,00,11]\}$, and $\LC[00, 11] = \{[5,00,11]\}$. 
Proposition~\ref{prop:MBN2_2fp_00_11_FP} below shows that all the MBNs built on BNs 
of $\FP[00,11]$ converge towards fixed points $(0,0)$ and $(\alpha,\beta)$ that are 
the only attractors. 

\begin{proposition}
	\label{prop:MBN2_2fp_00_11_FP}
	For any delay vector $dt$, every MBN built on a BN that belongs to $\FP[00,11]$ 
	admits only two attractors, fixed points $(0,0)$ and $(\alpha,\beta)$.
\end{proposition}

\begin{proof} 
	First, from Table~\ref{tab:FPs00-11}, it derives that network $[1,00,11]$ (resp. 
	$[9,00,11]$) is decreasing (resp. increasing). So, from 
	Proposition~\ref{prop:decreasing}, any MBN based on it (resp. on $[9,00,11]$) 
	converges towards its two fixed points $(0,0)$ and $(\alpha,\beta)$. It does so 
	in at most $\max(\alpha, \beta)$ time steps (resp. $1$ time step).\smallskip
	
	Furthermore, the interaction graph of network $[3,00,11]$ (resp. $[7,00,11]$ by 
	symmetry) is acyclic. So, from Proposition~\ref{prop:dag}, any MBN based on it 
	(resp. on $[7,00,11]$) converges towards the two fixed points and it does 
	so in at most $\beta$ (resp. $\alpha$) time steps.\smallskip
	
	Moreover, it is easy to see also that network $[6,00,11]$ (resp. $[8,00,11]$ by 
	symmetry) is a positive disjunctive BN. So, from Proposition~\ref{prop:disj_MBN}, 
	any MBN based on it (resp. based on $[8,00,11]$) converges towards the two fixed 
	points. It does so in at most $\beta+1$ (resp. $\alpha+1$ and $1$) time 
	steps.\smallskip
	
	Let us now focus on the MBNs built on BN $[2,00,11]$, given any delay vector 
	$dt$. We have: if $\rho = 0$ and $\gamma \geq 1$ then $(0, \gamma) 
	\xrightarrow{\gamma} (0,0)$\loopr{} that is a fixed point reached in $\gamma$ 
	time steps; if $\rho = 1$ and $\gamma = 0$ then $(1, 0) \to (0, \beta) 
	\xrightarrow{\beta} (0,0)$\loopr{} that is a fixed point reached in $\beta+1$ 
	time steps; if $\rho > 1$ and $\gamma = 0$ then $(\rho, 0) \to (\rho-1, \beta) 
	\to (\alpha, \beta)$\loopr{} that is a fixed point reached in $2$ time steps. So, 
	every MBN based on BN $[2,00,11]$ (resp. on BN $[4,00,11]$ by symmetry) admits 
	only two attractors, fixed points $(0,0)$ and $(\alpha,\beta)$, and its 
	convergence time is at most $\beta+1$ (resp. $\alpha+1$) time steps.
\end{proof}

Now, Proposition~\ref{prop:MBN2_2fp_00_11_LC} shows that there exist specific 
conditions under which MBNs built on BN $[5,00,11]$ of $\LC[00,11]$ evolve towards a 
limit cycle.

\begin{proposition}
	\label{prop:MBN2_2fp_00_11_LC}
	The only MBN built on BN $[5,00,11]$ that admits a limit cycle is set with 
	$\alpha = \beta = 1$, \ie BN $[5,00,11]$ itself. Any other MBN built $[5,00,11]$ 
	admits only two attractors, fixed points $(0,0)$ and $(\alpha,\beta)$.
\end{proposition}

\begin{proof}
	When $\alpha = \beta = 1$, the dynamics of this MBN is trivially the same as that 
	of BN $[5,00,11]$. The network admit three attractors, fixed points $(0,0)$ and 
	$(1,1)$ and limit cycle $(0,1) \leftrightarrows (1,0)$. Now, let us consider MBNs 
	such that $\alpha > 1$ or $\beta > 1$. We have: if $\alpha > 1$ and $\rho \geq 2$ 
	then $(\rho, 0) \to (\rho-1, \beta) \to (\alpha, \beta)$\loopr{}; if $\rho = 1$ 
	and $\beta > 1$ then $(1, 0) \to (0, \beta) \to (\alpha, \beta-1) \to (\alpha, 
	\beta)$\loopr{}; if $\beta > 1$ and $\gamma \geq 2$ then $(0,\gamma) \to (\alpha, 
	\gamma-1) \to (\alpha, \beta)$\loopr{}; if $\gamma = 1$ and $\alpha > 1$ then 
	$(0, 1) \to (\alpha, 0) \to (\alpha-1, \beta) \to (\alpha, \beta)$\loopr{}.
	Hence, only the MBN that is BN $[5,00,11]$ itself can admit a limit cycle of 
	length $2$. All the others converge towards $(0,0)$ and $(\alpha,\beta)$ in at 
	most $3$ time steps.
\end{proof}

\subsubsection{MBNs based on $\F[01, 10]$}
\label{sec:MBN2_2fp_01_10}

As $\F[00,11]$, the class of BNs $\F[01, 10]$ does not admit a symmetric class and 
what follows gives a characterization of the dynamics of MBNs built on it. 

\begin{table}[b!]
	{\scriptsize \centerline{
		$\begin{array}{|c||c|c|c|c|c|c|c|c|c|}
			\hline
			x & [1,01,10] & [2,01,10] & [3,01,10] & [4,01,10] & [5,01,10] & 
				[6,01,10] & [7,01,10] & [8,01,10] & [9,01,10]\\
			\hline\hline
			00 & 01 & 01 & 01 & 10 & 10 & 10 & 11 & 11 & 11\\
			01 & 01 & 01 & 01 & 01 & 01 & 01 & 01 & 01 & 01\\
			10 & 10 & 10 & 10 & 10 & 10 & 10 & 10 & 10 & 10\\
			11 & 00 & 01 & 10 & 00 & 01 & 10 & 00 & 01 & 10\\
			\hline
		\end{array}$
	}}
	\caption{Truth tables of all the $9$ BNs that admit fixed points $(0,1)$ and 
		$(1,0)$.}
	\label{tab:FPs01-10}
\end{table}

\begin{remark}
	\label{rem:MBN2_2fp_01_10}
	For every network of $\F[01,10]$, since $(0,1)$ and $(1,0)$ are fixed points, 
	whatever $\alpha$ and $\beta$ are, we have: if $\rho \geq 1$ and $\gamma = 0$ 
	then $(\rho, 0) \to (\alpha, 0)$\loopr{} that is a fixed point, and conversely, 
	if $\rho = 0$ and $\gamma \geq 1$ then $(0, \gamma) \to (0, \beta)$\loopr{} that 
	is a fixed point.
\end{remark}
Therefore we analyze the dynamical behavior of all the initial configurations of the 
form $(0, 0)$ and $(\rho, \gamma)$ where $\rho, \gamma \neq 0$. From 
Table~\ref{tab:FPs01-10}, a basic enumeration gives that $\FP[01,10] = \F[01,10] 
\setminus \{[7,01,10]\}$, and $\LC[01,10] = \{[7,01,10]\}$. 
Proposition~\ref{prop:MBN2_2fp_01_10_FP} below shows that all the MBNs built on BNs 
of $\FP[01,10]$ converge towards fixed points $(0,\beta)$ and $(\alpha,0)$ that are 
the only attractors. 

\begin{proposition}
	\label{prop:MBN2_2fp_01_10_FP}
	For any delay vector $dt$, every MBN built on a BN that belongs to $\FP[01,10]$ 
	admits only two attractors, fixed points $(0,\beta)$ and $(\alpha,0)$.
\end{proposition}

\begin{proof}
	First of all, let us focus on BN $[3,01,10]$. Its interaction graph does not 
	induce cycles except one positive loop. So, by Proposition~\ref{prop:dag}, any 
	MBN built on it only admit fixed points. More precisely, whatever $\alpha$ and 
	$\beta$ are, configuration $(0,0)$ converges towards $(0, \beta)$ in $1$ time 
	step. Any other configuration such that $\rho, \gamma \geq 1$ admits the 
	following trajectory: $(\rho, \gamma) \to (\alpha, \gamma-1) 
	\xrightarrow{\gamma-1} (\alpha, 0)$\loopr{} that is reached in $\gamma$ time 
	steps. Thus, any MBN built on BN $[3,01,10]$ (resp. on BN $[5,01,10]$ by 
	symmetry) converges towards two attractors, fixed points $(0, \beta)$ and 
	$(\alpha, 0)$, and does so in at most $\beta$ (resp. $\alpha$) time 
	steps.\smallskip

	Now, consider BN $[1,01,10]$. If $\rho > \gamma$ then the trajectory is $(\rho, 
	\gamma) \xrightarrow{\gamma} (\rho - \gamma, 0) \to (\alpha, 0)$\loopr{} that is 
	reached in $\gamma+1$ time steps. Conversely, if $\rho \leq \gamma$ then the 
	trajectory is $(\rho, \gamma) \xrightarrow{\rho} (0, \gamma - \rho) \to 
	(0, \beta)$ that is reached in $\rho+1$ time steps. Thus, any MBN built on BN 
	$[1,01,10]$ (resp. on BN $[4,01,10]$ by symmetry) converges towards two 
	attractors, fixed points $(0,\beta)$ and $(\alpha, 0)$, and does so in at most 
	$\max(\alpha, \beta) + 1$ time steps.\smallskip
		
	Concerning BN $[2,01,10]$, for all delay vector $dt$, $(0,0)$ converges towards 
	$(0,\beta)$ in one time step. Any other initial configuration $(\rho, \gamma)$ 
	such that $\rho, \gamma \geq 1$ admits the following trajectory: $(\rho, \gamma) 
	\to (\rho-1, \beta) \xrightarrow{\rho-1} (0,\beta)$\loopr{} that is reached in at 
	most $\rho$ time steps. Thus, any MBN built on BN $[2,01,10]$ (resp. on BN 
	$[6,01,10]$ by symmetry) converges towards two attractors, fixed points $(0,
	\beta)$ and $(\alpha, 0)$, and does so in at most $\alpha$ (resp. $\beta$) time 
	steps.\smallskip
		
	Lastly, concerning BN $[8,01,10]$, for all delay vector $dt$, $(0,0)$ admits the 
	following trajectory: $(0,0) \to (\alpha, \beta) \xrightarrow{\alpha} (0, 
	\beta)$\loopr{} that is reached in $\alpha+1$ time steps. Any other initial 
	configuration $(\rho, \gamma)$ such that $\rho, \gamma \geq 1$ admits the 
	following trajectory: $(\rho, \gamma) \to (\rho-1, \beta) \xrightarrow{\rho-1} 
	(0,\beta)$\loopr{} that is reached in $\rho$ time steps. Thus, any MBN built on 
	BN $[8,01,10]$ (resp. on BN $[9,01,10]$ by symmetry) converges towards two 
	attractors, fixed points $(0,\beta)$ and $(\alpha, 0)$, and does so in at most 
	$\alpha+1$ (resp. $\beta+1$) time steps.
\end{proof}

Now, Proposition~\ref{prop:MBN2_2fp_01_10_LC} shows that there exist specific 
conditions under which MBNs built on BN $[7,01,10]$ of $\LC[01,10]$ evolve towards a 
limit cycle.

\begin{proposition}
	\label{prop:MBN2_2fp_01_10_LC}
	Every MBN built on BN $[7,01,10]$ with $\alpha = \beta$ admits a limit cycle. Any 
	other admits only $(0,\beta)$ and $(\alpha,0)$ as its unique attractors.
\end{proposition}

\begin{proof}
	Consider first that $\alpha = \beta$ and that $\rho = \gamma$. Then we have: 
	$(\rho, \rho) \xrightarrow{\rho} (0,0) \to (\alpha, \alpha) \xrightarrow{\alpha} 
	(0,0)$, which emphasizes a limit cycle of length $\alpha$. Now, consider that 
	$\alpha \neq \beta$ and that $\rho = \gamma$. We have the following trajectory: 	
	if $\alpha > \beta$ then $(\rho,\rho) \xrightarrow{\rho} (0,0) \to (\alpha, 
	\beta) \xrightarrow{\beta} (\alpha - \beta, 0) \to (\alpha, 0)$\loopr{} that is 
	reached in $\rho + \beta + 2$ time steps; if $\alpha < \beta$ then $(\rho,\rho) 
	\xrightarrow{\rho} (0,0) \to (\alpha, \beta) \xrightarrow{\alpha} (0, \beta - 
	\alpha) \to (0, \beta)$\loopr{} that is reached in $\rho + \alpha + 2$ time 
	steps. The last case to consider is when $\rho \neq \gamma$ whatever $\alpha$ and 
	$\beta$, for which we have: if $\rho > \gamma$ then $(\rho, \gamma) 
	\xrightarrow{\gamma} (\rho - \gamma, 0) \to (\alpha, 0)$\loopr{} that is reached 
	in $\gamma+1$ time steps; if $\rho < \gamma$ then $(\rho, \gamma) 
	\xrightarrow{\rho} (0, \gamma - \rho) \to (0, \beta)$\loopr{} that is reached in 
	$\rho+1$ time steps. Hence, when built on BN $[7,01,10]$, only MBNs such that 
	$\alpha = \beta$ admit a limit cycle of length $\alpha$ that can only be reached 
	by configurations such that $\rho = \gamma$. All the others converge only towards 
	$(0,0)$ and $(\alpha, \beta)$ in at most $\alpha + \max(\alpha,\beta) + 2$ time 
	steps.
\end{proof}

\subsubsection{MBNs based on $\F[01, 11]$ and $\F[10, 11]$}
\label{sec:MBN2_2fp_01_11}

\begin{table}[t!]
	{\scriptsize \centerline{
		$\begin{array}{|c||c|c|c|c|c|c|c|c|c|}
			\hline
			x & [1,01,11] & [2,01,11] & [3,01,11] & [4,01,11] & [5,01,11] & 
				[6,01,11] & [7,01,11] & [8,01,11] & [9,01,11]\\
			\hline\hline
			00 & 01 & 01 & 01 & 10 & 10 & 10 & 11 & 11 & 11\\
			01 & 01 & 01 & 01 & 01 & 01 & 01 & 01 & 01 & 01\\
			10 & 00 & 01 & 11 & 00 & 01 & 11 & 00 & 01 & 11\\
			11 & 11 & 11 & 11 & 11 & 11 & 11 & 11 & 11 & 11\\
			\hline
		\end{array}$}}\smallskip
		
	{\scriptsize \hspace*{1pt}\centerline{
		$\begin{array}{|c||c|c|c|c|c|c|c|c|c|}
			\hline
			x & [1,10,11] & [2,10,11] & [3,10,11] & [4,10,11] & [5,10,11] & 
				[6,10,11] & [7,10,11] & [8,10,11] & [9,10,11]\\
			\hline\hline
			00 & 01 & 01 & 01 & 10 & 10 & 10 & 11 & 11 & 11\\
			01 & 00 & 10 & 11 & 00 & 10 & 11 & 00 & 10 & 11\\
			10 & 10 & 10 & 10 & 10 & 10 & 10 & 10 & 10 & 10\\
			11 & 11 & 11 & 11 & 11 & 11 & 11 & 11 & 11 & 11\\
			\hline
		\end{array}$}}
	\caption{Truth tables of all the $9$ BNs that admit (up) fixed points $(0,1)$ and 
		$(1,1)$ and (down) fixed points $(1,0)$ and $(1,1)$.}
	\label{tab:FPs01-11}
\end{table}

Concerning the classes of BNs $\F[01, 11]$ and $\F[10, 11]$, notice first that they 
are symmetric. Thus, the results obtained for $\F[01, 11]$ have their symmetric that 
hold for $\F[10, 11]$. So, let us focus only on $\F[01, 11]$.

\begin{remark}
	\label{rem:MBN2_2fp_01_11}
	For every network of $\F[01, 11]$, since $(0,1)$ and $(1,1)$ are fixed points, 
	whatever $\alpha$ and $\beta$ are, configurations $(0,\gamma)$, with $\gamma \geq 
	1$, converge towards $(0,\beta)$ in one time step, and any initial configuration 
	such that $\rho \geq 1$ and $\gamma \geq 1$ converges towards $(\alpha,\beta)$ in 
	one time step. 
\end{remark}
Therefore we analyze the dynamical behavior of all the initial configurations of the 
form $(\rho, 0)$, where $\rho \geq 0$. From Table~\ref{tab:FPs01-11} (up), a basic 
enumeration gives that $\FP[00, 01] = \F[00,01] \setminus \{[4,00,01]\}$, and 
$\LC[00, 01] = \{[4,00,01]\}$. Proposition~\ref{prop:MBN2_2fp_01_11_FP} below shows 
that all the MBNs built on BNs of $\FP[01,11]$ converge towards fixed points $(0,
\beta)$ and $(\alpha,\beta)$ that are the only attractors. 

\begin{proposition}
	\label{prop:MBN2_2fp_01_11_FP}
	For any delay vector $dt$, every MBN built on a BN that belongs to $\FP[01,11]$ 
	admits only two attractors, fixed points $(0,\beta)$ and $(\alpha,\beta)$.
\end{proposition}

\begin{proof}
	First of all, consider BNs $[3,01,11]$, $[6,01,11]$ and $[9,01,11]$ that are 
	increasing according to their definition. So, by 
	Proposition~\ref{prop:decreasing}, all of the MBNs built on them admit only fixed 
	points.
		
	Concerning MBNs built on $[3,01,11]$, whatever $\alpha$ and $\beta$ are, 
	configuration $(0,0)$ converges towards $(0, \beta)$ in $1$ time step. Any other 
	configuration such that $(\rho, 0)$, with $\rho \geq 1$, admits the following 
	trajectory: $(\rho, 0) \to (\alpha, \beta)$\loopr{} that is reached in $1$ time 
	step. Thus, any MBN built on BN $[3,01,11]$ admits only two attractors, fixed 
	points $(0, \beta)$ and $(\alpha, \beta)$, and converge to them in at most $2$ 
	time steps.
	
	Concerning MBNs built on $[6,01,11]$, whatever $\alpha$ and $\beta$ are, 
	configuration $(0,0)$ becomes $(\alpha, 0)$ in one time step. Moreover, any 
	configuration $(\rho, 0)$, with $1 \leq \rho \leq \alpha$, converges towards 
	$(\alpha, \beta)$ in one time step. Thus, any MBN built on BN $[6,01,11]$ admits 
	only two attractors, fixed points $(0, \beta)$ and $(\alpha, \beta)$, and 
	converge to them in at most $2$ time steps.\smallskip
	
	Concerning MBNs built on $[9,01,11]$, whatever $\alpha$ and $\beta$ are, 
	any configuration such that $(\rho, 0)$, with $\rho \geq 0$ converges towards 
	$(\alpha, \beta)$ in $1$ time step. So, such networks admit only two attractors, 
	fixed points $(0, \beta)$ and $(\alpha, \beta)$, and converge to them in at most 
	$1$ time step.\smallskip
	
	Now, let us focus on BN $[2,01,11]$. Its interaction graph does not induce cycles 
	except one positive loop. So, by Proposition~\ref{prop:dag}, any MBN built on 
	such a BN only admits fixed points. More precisely, whatever $\alpha$ and $\beta$ 
	are, configurations $(0,0)$ and $(1,0)$ converge towards $(0, \beta)$ in $1$ time 
	step. Any other configuration such that $(\rho, 0)$, with $\rho \geq 2$, admits 
	the following trajectory: $(\rho, 0) \to (\rho-1, \beta) \to (\alpha, 
	\beta)$\loopr{} that is reached in $2$ time steps. Thus, any MBN built on BN 
	$[2,01,11]$ admits only two attractors, fixed points $(0, \beta)$ and $(\alpha, 
	\beta)$, and converge to them in at most $2$ time steps.\smallskip
		
	Let us focus on BN $[1,01,11]$. Whatever $\alpha$ and $\beta$ are, 
	any configuration $(\rho,0)$, with $\rho \geq 0$, admits the following 
	trajectory: $(\rho, 0) \xrightarrow{\rho} (0,0) \to (0, \beta)$\loopr{} that is 
	reached in $\rho+1$ time steps. Thus, any MBN built on BN $[1,01,11]$ admits only 
	two attractors, fixed points $(0, \beta)$ and $(\alpha, \beta)$, and converge to 
	them in at most $\alpha+1$ time steps.\smallskip
	
	Consider BN $[5,01,11]$. Whatever $\alpha$ and $\beta$ are, configuration 
	$(1,0)$ converges towards $(0,\beta)$ in $1$ time step. Configuration $(0,0)$ 
	follows the trajectory $(0,0) \to (\alpha, 0) \to (\alpha-1, \beta) \to 
	(\alpha, \beta)$\loopr{} that is reached in $3$ time steps. Any other 
	configuration such that $(\rho, 0)$, with $\rho \geq 2$, admits the following 
	trajectory: $(\rho, 0) \to (\rho-1, \beta) \to (\alpha,\beta)$\loopr{} thay is 
	reached in $2$ time steps. Thus, any MBN built on BN $[5,01,11]$ admits only two 
	attractors, fixed points $(0, \beta)$ and $(\alpha, \beta)$, and converge to them 
	in at most $3$ time steps.\smallskip
	
	Consider BN $[7,01,11]$. Whatever $\alpha$ and $\beta$ are, any configuration 
	$(\rho, 0)$, with $\rho \geq 0$, admits the following trajectory: $(\rho, 0) 
	\xrightarrow{\rho} (0,0) \to (\alpha, \beta)$\loopr{} that is reached in $\rho+1$ 
	time steps. Thus, any MBN built on BN $[7,01,11]$ admits only two attractors, 
	fixed points $(0, \beta)$ and $(\alpha, \beta)$, and converge to them 
	in at most $\alpha+1$ time steps.\smallskip
	
	Let us finally focus on BN $[8,01,11]$. Whatever $\alpha$ and $\beta$ are, 
	configuration $(0,0)$ (resp. $(1,0)$) converges towards $(\alpha, \beta)$ (resp. 
	$(0,\beta)$) in $1$ time step. Any other configuration such that $(\rho, 0)$, 
	with  $\rho \geq 2$, admits the following trajectory: $(\rho, 0) \to 
	(\rho-1, \beta) \to (\alpha, \beta)$\loopr{} that is reached in $2$ time steps. 
	Thus, any MBN built on BN $[8,01,11]$ admits only two attractors, fixed points 
	$(0, \beta)$ and $(\alpha, \beta)$, and converge to them in at most $2$ time 
	steps.
\end{proof}

Now, Proposition~\ref{prop:MBN2_2fp_01_11_LC} shows that there exist specific 
conditions under which MBNs built on BN $[4,01,11]$ of $\LC[01,11]$ evolve towards a 
limit cycle.

\begin{proposition}
	\label{prop:MBN2_2fp_01_11_LC}
	Every MBN built on BN $[4,01,10]$ admits a limit cycle that is reached by any 
	configuration $(\rho, 0)$, with $0 \leq \rho \leq \alpha$.
\end{proposition}

\begin{proof}
	First, let us consider any configuration $(\rho, 0)$, with $\rho \geq 0$. Such a 
	configuration follows the trajectory $(\rho,0) \xrightarrow{\rho} (0,0) \to 
	(\alpha, 0) \to (\alpha-1, 0) \xrightarrow{\alpha-1} (0,0)$, which emphasizes a 
	limit cycle of length $\alpha$ composed of all the configurations $(\rho, 0)$. 
	Thus, any MBN built on BN $[4,01,10]$ admits three attractors, the two fixed 
	points $(0, \beta)$ and $(\alpha, \beta)$  that are reached in $1$ time step, and 
	a limit cycle of length $\alpha$.
\end{proof}

\subsection{Networks admitting three fixed points}
\label{sec:MBN2_3fp}

In this section, we focus on the MBNs that can be built on the basis of BNs of size 
$2$ that admit three fixed points. First of all, let us notice that there exist $4$ 
distinct classes of such networks, each of which being composed of $3$ networks. As 
what has been presented above, let us introduce the following notations. 

\begin{notation}
	$[k, x, y, z]$, with $k \in \{1, \dots, 3\}$ and $x, y, z \in \B^2$, denotes the 
	network of size $2$ whose local transition functions are represented by $k$ and 
	defined by their truth tables in Tables~\ref{tab:FPs00-01-10_01-10-11} 
	and~\ref{tab:FPs00-01-11_00-10-11}, and that admits $x$, $y$ and $z$ as its fixed 
	points (represented as binary words).
\end{notation}

Let us also denote by $\F[x, y, z]$ the set composed of all BNs that admit $x$, $y$ 
and $z$ as their unique fixed points.

\subsubsection*{MBNs based on $\F[00,01,10]$ and $\F[01,10,11]$}
\label{sec:MBN2_3fp_00_01_10}

Concerning the classes of BNs $\F[00,01,10]$ and $\F[01,10,11]$, notice first that 
they are symmetric. Thus, all the results obtained for $\F[00,01,10]$ have their 
symmetric that hold for $\F[01,10,11]$. So, let us focus only on $\F[00,01,10]$.

\begin{remark}
	For every network of $\F[00,01,10]$, since $(0,0)$, $(0,1)$ and $(1,0)$ are fixed 
	points, whatever $\alpha$ and $\beta$ are, we have: if $\rho = \gamma = 0$ then 
	there is trivially convergence towards fixed point $(0,0)$; if $\rho = 0$ and 
	$\gamma \geq 1$ then $(0, \gamma) \to (0, \beta)$\loopr{} that is a fixed point; 
	if $\rho \geq 1$ and $\gamma = 0$ then $(\rho, 0) \to (\alpha, 0)$\loopr{} that 
	is a fixed point. 
\end{remark}
Therefore, we analyze the dynamical behavior of all the initial configurations of the 
form $(\rho, \gamma)$, where $\rho \geq 1$ and $\gamma \geq 1$.
Proposition~\ref{prop:MBN2_3fp_00_01_10_FP} below shows that all the MBNs built on 
BNs of $\F[00,01,10]$ converge towards fixed points $(0,0)$, $(0,\beta)$ and 
$(\alpha,0)$ that are the only attractors. 

\begin{table}[h!]
	{\scriptsize \begin{center}
		\begin{minipage}{.475\textwidth}
			\centerline{
				$\begin{array}{|c||c|c|c|}
					\hline
					x & [1,00,01,10] & [2,00,01,10] & [3,00,01,10]\\
					\hline\hline
					00 & 00 & 00 & 00\\
					01 & 01 & 01 & 01\\
					10 & 10 & 10 & 10\\
					11 & 00 & 01 & 10\\
					\hline
				\end{array}$
			}
		\end{minipage}\quad
		\begin{minipage}{.475\textwidth}
			\centerline{
				$\begin{array}{|c||c|c|c|}
					\hline
					x & [1,01,10,11] & [2,01,10,11] & [3,01,10,11]\\
					\hline\hline
					00 & 01 & 10 & 11\\
					01 & 01 & 01 & 01\\
					10 & 10 & 10 & 10\\
					11 & 11 & 11 & 11\\
					\hline
				\end{array}$
			}
		\end{minipage}\medskip
	
		\begin{minipage}{.475\textwidth}
			\centerline{(a)}
		\end{minipage}\quad
		\begin{minipage}{.475\textwidth}
			\centerline{(b)}
		\end{minipage}
	\end{center}}
		
	\caption{(a) Truth tables of all the $3$ BNs that admit fixed points $(0,0)$, 
		$(0,1)$ and $(1,0)$; (b) Truth tables of all the $3$ BNs that admit fixed 
		points $(0,1)$, $(1,0)$ and $(1,1)$.}
	\label{tab:FPs00-01-10_01-10-11}
\end{table}

\begin{proposition}
	\label{prop:MBN2_3fp_00_01_10_FP}
	For any delay vector $dt$, every MBN built on a BN that belongs to $\F[00,01,10]$ 
	admits only three attractors, fixed points $(0,0)$, $(0,\beta)$ and 
	$(\alpha, 0)$.
\end{proposition}

\begin{proof}
	Consider first BN $[1,00,01,10]$. Whatever $\alpha$ and $\beta$ are, given $\rho 
	\geq 1$ and $\gamma \geq 1$, we have: if $\rho < \gamma$ then $(\rho, \gamma) 
	\xrightarrow{\rho} (0, \gamma-\rho) \to (0, \beta)$\loopr{} that is reached in 
	$\rho+1$ time steps; if $\rho > \gamma$ then $(\rho, \gamma) \xrightarrow{\gamma} 
	(\rho-\gamma, 0) \to (\alpha, 0)$\loopr{} that is reached in $\gamma+1$ time 
	steps; if $\rho = \gamma$ then $(\rho, \rho) \xrightarrow{\rho} (0, 0)$\loopr{} 
	that is reached in $\rho$ time steps. Thus, any MBN built on BN $[1,00,01,10]$ 
	admits three attractors, fixed points $(0,0)$, $(0, \beta)$ and $(\alpha, 0)$  
	that are reached in at most $\max(\alpha, \beta) + 1$ time steps.\smallskip
	
	Concerning BN $[2,00,01,10]$, whatever $\alpha$ and $\beta$ are, any initial 
	configuration such that $\rho \geq 1$ and $\gamma \geq 1$ admits the following 
	trajectory: $(\rho, \gamma) \to (\rho-1, \beta) \xrightarrow{\rho-1} (0, 
	\beta)$\loopr{} that is reached in $\rho$ time steps. Thus, any MBN built on BN 
	$[2,00,01,10]$ (resp. on BN $[3,00,01,10]$ by symmetry) admits three attractors, 
	fixed points $(0,0)$, $(0, \beta)$ and $(\alpha, 0)$, and converges to them in at 
	most $\alpha$ (resp. $\beta$) time steps.
\end{proof}

\subsubsection*{MBNs based on $\F[00,01,11]$ and $\F[00,10,11]$}
\label{sec:MBN2_3fp_00_01_11}

Concerning the classes of BNs $\F[00,01,11]$ and $\F[00,10,11]$, notice first that 
they are symmetric. Thus, all the results obtained for $\F[00,01,11]$ have their 
symmetric that hold for $\F[00,10,11]$. So, let us focus only on $\F[00,01,11]$.

\begin{remark}
	For every network of $\F[00,01,11]$, since $(0,0)$, $(0,1)$ and $(1,1)$ are fixed 
	points, whatever $\alpha$ and $\beta$ are, we have: if $\rho = \gamma = 0$ then 
	there is trivially convergence towards fixed point $(0,0)$; if $\rho = 0$ and 
	$\gamma \geq 1$ then $(0, \gamma) \to (0, \beta)$\loopr{} that is a fixed point; 
	if $\rho \geq 1$ and $\gamma \geq 1$ then $(\rho, \gamma) \to (\alpha, 
	\beta)$\loopr{} that is a fixed point. 
\end{remark}
Therefore, we analyze the dynamical behavior of all the initial configurations of the 
form $(\rho, \gamma)$, where $\rho \geq 1$ and $\gamma = 0$. 
Proposition~\ref{prop:MBN2_3fp_00_01_11_FP} below shows that all the MBNs built on 
BNs of $\F[00,01,11]$ converge towards fixed points $(0,0)$, $(0,\beta)$ and 
$(\alpha,\beta)$ that are the only attractors. 

\begin{proposition}
	\label{prop:MBN2_3fp_00_01_11_FP}
	For any delay vector $dt$, every MBN built on a BN that belongs to $\F[00,01,11]$ 
	admits only three attractors, fixed points $(0,0)$, $(0,\beta)$ and 
	$(\alpha, \beta)$.
\end{proposition}

\begin{proof}
	Consider first BN $[1,00,01,11]$. Whatever $\alpha$ and $\beta$ are, given $\rho 
	\geq 1$ and $\gamma = 0$, configuration $(\rho, 0)$ admits the following 
	trajectory: $(\rho, 0) \to (\rho-1, 0) \xrightarrow{\rho-1} (0,0)$\loopr{} that 
	is reached in $\rho$ time steps. Thus, any MBN built on BN $[1,00,01,11]$ admits 
	three attractors, fixed points $(0,0)$, $(0, \beta)$ and $(\alpha, \beta)$ that 
	are reached in at most $\alpha$ time steps.\smallskip
	
	Concerning BN $[2,00,01,11]$, whatever $\alpha$ and $\beta$ are, given $\rho 
	\geq 1$ and $\gamma = 0$, we have: if $\rho = 1$ then $(1, 0) \to (0, 
	\beta)$\loopr{} that is reached in $1$ time step; if $\rho > 1$ then $(\rho, 0) 
	\to (\rho-1, \beta) \to (\alpha, \beta)$\loopr{} that is reached in $2$ time 
	steps. Thus, any MBN built on BN $[2,00,01,11]$ admits three attractors, fixed 
	points $(0,0)$, $(0, \beta)$ and $(\alpha, \beta)$ that are reached in at most 
	$2$ time steps.\smallskip
	
	Concerning BN $[3,00,01,11]$, whatever $\alpha$ and $\beta$ are, given $\rho 
	\geq 1$ and $\gamma = 0$, any configuration $(\rho, 0)$ converges towards 
	$(\alpha, \beta)$ in $1$ time step.
\end{proof}

\begin{table}[h!]
	{\scriptsize \begin{center}
		\begin{minipage}{.475\textwidth}
			\centerline{
				$\begin{array}{|c||c|c|c|}
					\hline
					x & [1,00,01,11] & [2,00,01,11] & [3,00,01,11]\\
					\hline\hline
					00 & 00 & 00 & 00\\
					01 & 01 & 01 & 01\\
					10 & 00 & 01 & 11\\
					11 & 11 & 11 & 11\\
					\hline
				\end{array}$
			}
		\end{minipage}\quad
		\begin{minipage}{.475\textwidth}
			\centerline{
				$\begin{array}{|c||c|c|c|}
					\hline
					x & [1,00,10,11] & [2,00,10,11] & [3,00,10,11]\\
					\hline\hline
					00 & 00 & 00 & 00\\
					01 & 00 & 10 & 11\\
					10 & 10 & 10 & 10\\
					11 & 11 & 11 & 11\\
					\hline
				\end{array}$
			}
		\end{minipage}\medskip
	
		\begin{minipage}{.475\textwidth}
			\centerline{(a)}
		\end{minipage}\quad
		\begin{minipage}{.475\textwidth}
			\centerline{(b)}
		\end{minipage}
	\end{center}}
		
	\caption{(a) Truth tables of all the $3$ BNs that admit fixed 
		points $(0,0)$, $(0,1)$ and $(1,1)$; (b) Truth tables of all the $3$ BNs that 
		admit fixed points $(0,0)$, $(1,0)$ and $(1,1)$.}
	\label{tab:FPs00-01-11_00-10-11}
\end{table}


To conclude on this theoretical analysis, let us simply add that it is trivial to 
show that MBNs of size $2$ that are built on the only BN with $4$ fixed points cannot 
admit limit cycle. Their attractors are fixed points $(0,0)$, $(\alpha, 0)$, $(0, 
\beta)$ and $(\alpha, \beta)$. Furthermore, as it could have been predicted, the more 
BNs admit degree of freedom (\ie the less they admit fixed points), the more MBNs 
built on them may have complex behaviors. Eventually, as it has been highlighted 
in~\cite{Graudenzi2011a,Graudenzi2011b}, this section has formally shown that, for 
networks of size $2$, higher delays results in longer limit cycles. However, it has 
also been shown that it is not true in general that the more the maximum delay value, 
the less the network admits attractors. Indeed, we have seen that this property holds 
in some cases but that there exist also networks for which increasing delays can 
create limit cycles asymptotically.

\section{Application to specific genetic regulation networks}
\label{sec:applis}

\subsection{Immunity control in bacteriophage $\lambda$}
\label{sec:Thomas}

\begin{figure}[b!]
	\begin{center}
		\begin{minipage}{.375\textwidth}
			\scalebox{.85}{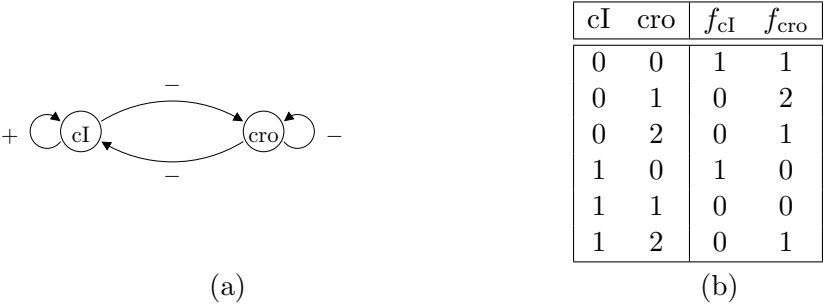}
		\end{minipage}
		\qquad\qquad
		\begin{minipage}{.25\textwidth}
			\begin{tabular}{|cc|cc|}
				\hline
				$\gcI$ & $\gcro$ & $f_{\gcI}$ & $f_{\gcro}$\\
				\hline\hline
				0 & 0 & 1 & 1\\
				0 & 1 & 0 & 2\\
				0 & 2 & 0 & 1\\
				1 & 0 & 1 & 0\\
				1 & 1 & 0 & 0\\
				1 & 2 & 0 & 1\\
				\hline
			\end{tabular}
		\end{minipage}\smallskip
			
		\begin{minipage}{.375\textwidth}
			\centerline{(a)}
		\end{minipage}
		\qquad\qquad
		\begin{minipage}{.25\textwidth}
			\centerline{(b)}
		\end{minipage}
	\end{center}
	\caption{(a) Interaction graph of the gene regulation network implying genes 
		$\gcI$ and $\gcro$ in the immunity control of the bacteriophage 
		$\lambda$ introduced in~\cite{Thomas1995}; (b) Table presenting the 
		dynamical behavior of this network inferred from the phase diagram 
		given by Thieffry and Thomas.}
	\label{fig_lambda_Thomas}
\end{figure}

In~\cite{Thomas1995}, Thieffry and Thomas proposed a logical model based on the 
Thomas' method~\cite{Thomas1973,Thomas1988,Thomas1991} in order to analyze and 
achieve a better understanding of the role that specific genes have in the decision 
between lysis and lysogenization in bacteriophage $\lambda$. Notably, they first 
introduced an interaction graph, composed of four genes $\gcI$, $\gcII$, $\gcro$ 
and $\gN$, that they voluntarily simplified in a two-genes model in order to 
focus especially on the interactions existing between $\gcI$ and $\gcro$. 
Without entering neither into the details of the Thomas' method nor into those 
of the model itself (they can be obtained in the original paper), let us just 
give in Figure~\ref{fig_lambda_Thomas} the main static and dynamical features of 
the latter, in which we can see that the chosen modeling is not in the Boolean 
setting but in a discrete one. Indeed, although gene $\gcI$ is associated with a 
Boolean variable, gene $\gcro$ is associated with a variable taking values into 
$\{0,1,2\}$. In this modeling, the network converges towards two attractors, a 
stable configuration $(1,0)$, whose corresponding expression pattern is $\gcI$ 
expressed and $\gcro$ inhibited that stands for the bacterium lysing, and a 
stable oscillation $(0,1) \leftrightarrows (0,2)$ that stand for the bacterium 
becoming lysogenic.
				
From this model, considering that when gene $\gcro$ is at its maximum expression 
level, it inevitably tends to decrease to its medium expression level whatever 
that of $\gcI$, it is trivial to reduce this model into the Boolean setting by 
merging states $1$ and $2$ of $\gcro$, without loss of qualitative information 
from the biological standpoint because two attractors are conserved exactly, 
through the two stable configurations $(1,0)$ and $(0,1)$. Hence, from 
Figure~\ref{fig_lambda_Thomas}.b, it is easy to obtain the truth table of the 
Boolean local transition functions of $\gcI$ and $\gcro$ and consequently the 
associated BN. Then, we can obviously build the related interaction and 
transitions graphs (see Figure~\ref{fig_lambda_Boole}). 
				
\begin{figure}[t!]
	\begin{center}
		\begin{minipage}{.25\textwidth}
			\begin{tabular}{|cc|cc|}
				\hline
				$\gcI$ & $\gcro$ & $f_{\gcI}$ & $f_{\gcro}$\\
				\hline\hline
				0 & 0 & 1 & 1\\
				0 & 1 & 0 & 1\\
				1 & 0 & 1 & 0\\
				1 & 1 & 0 & 0\\
				\hline
			\end{tabular}
		\end{minipage}
		\qquad\quad
		\begin{minipage}{.25\textwidth}
			\scalebox{.85}{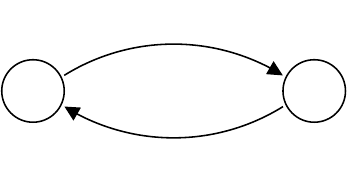}
		\end{minipage}
		\qquad\quad
		\begin{minipage}{.23\textwidth}
			\scalebox{.85}{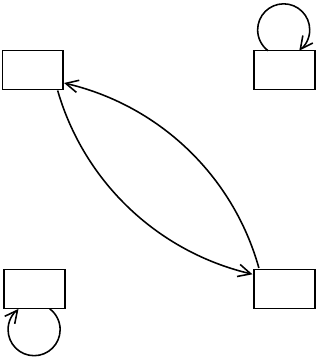}
		\end{minipage}\smallskip
		
		\begin{minipage}{.25\textwidth}
			\centerline{(a)}
		\end{minipage}
		\qquad\quad
		\begin{minipage}{.25\textwidth}
			\centerline{(b)}
		\end{minipage}
		\qquad\quad
		\begin{minipage}{.23\textwidth}
			\centerline{(c)}
		\end{minipage}
	\end{center}
	\caption{(a) Truth table of the Boolean local transition functions of $\gcI$ 
		and $\gcro$ inferred from the Thomas' model of the immunity control of 
		the bacteriophage $\lambda$; (b) Related interaction graph; (c) Related 
		transition graph.}
	\label{fig_lambda_Boole}
\end{figure}
												
Now, considering this network set in the MBN context, we obtain 
Theorem~\ref{thm:Lambda} below.
												
\begin{theorem}
	\label{thm:Lambda}
	Any MBN built on the BN of the immunity control of the bacteriophage 
	$\lambda$ admits fixed points $(dt_\gcI, 0)$ and $(0, dt_\gcro)$. Only those 
	such that $dt_\gcI = dt_\gcro$ can admit a limit cycle of length $dt_\gcI$.
\end{theorem}
		
\begin{proof}
	First, for all $dt_\gcI$ and $dt_\gcro$ and initial configurations such that 
	$\rho \neq \gamma$, we have: if $\rho > \gamma$ then $(\rho, \gamma) 
	\xrightarrow{\gamma} (\rho-\gamma, 0) \to (dt_\gcI, 0)$\loopr{} that is reached 
	in at most $dt_\gcro$ time steps, and if $\rho < \gamma$ then $(\rho, \gamma) 
	\xrightarrow{\rho} (0, \gamma-\rho) \to (0, dt_\gcro)$\loopr{} that is reached in 
	at most $dt_\gcI$ time steps. Now, let us consider that $\rho = \gamma$ and 
	compute its trajectory depending on $dt_\gcI$ and $dt_\gcro$:
	\begin{itemize}
	\item if $dt_\gcI = dt_\gcro$ then $(\rho, \rho) \to (\rho-1, \rho-1) 
		\xrightarrow{\rho-1} (0,0) \to (\rho, \rho)$, which emphasizes a limit cycle 
		of length $dt_\gcI$.
	\item if $dt_\gcI > dt_\gcro$ then $(\rho, \rho) \xrightarrow{\rho} (0,0) \to 
		(dt_\gcI, dt_\gcro) \xrightarrow{dt_\gcro} (dt_\gcI - dt_\gcro, 0) \to 
		(dt_\gcI, 0)$\loopr{} that is reached in at most $2dt_\gcro+2$ time steps.
	\item if $dt_\gcI < dt_\gcro$ then $(\rho, \rho) \xrightarrow{\rho} (0,0) \to 
		(dt_\gcI, dt_\gcro) \xrightarrow{dt_\gcro} (0, dt_\gcro-dt_\gcI) \to 
		(0, dt_\gcro)$\loopr{} that is reached in at most $2dt_\gcI+2$ time steps.
	\end{itemize}
\end{proof}

\begin{remark}
	\label{rem:synchro}
	As Theorem~\ref{thm:Lambda} states, in the case where $dt_\gcI = dt_\gcro$ and  
	$\alpha = \beta$, the behavior of the network evolves towards a limit cycle. This 
	limit cycle is a spurious asymptotic behavior induced by the perfect 
	synchronicity created by both the parallel updating mode and the equality of all 
	the control parameters $dt_\gcI$, $dt_\gcro$, $\alpha$ and $\beta$. These 
	parameters could be useful to control synchronicity loss in genetic regulatory 
	networks. 
\end{remark}

\subsection{Floral morphogenesis of \emph{Arabidopsis thaliana}}
\label{sec:Mendoza}
												
\setcounter{MaxMatrixCols}{20}
In~\cite{Mendoza1998}, Mendoza and Alvarez-Buylla introduced a network of the genetic 
control of the floral morphogenesis of the plant \emph{Arabidopsis thaliana}. The 
mathematical model that they used to study the dynamical behavior of this network was 
that of threshold BNs. Threshold BNs are particular BNs in which every local 
transition function $f_i(x)$ is a threshold Boolean function defined as:
\begin{equation*}
	f_i(x) = \mathbb{I}(\sum_{j \in V} w_{i,j} \cdot x_j - \theta_i)\text{,}
\end{equation*}
where $\theta_i$ is the activation threshold of node $i$, $w_{i,j} \in \mathbb{R}$ is 
the interaction weight that node $j$ has on node $i$, and $\mathbb{I}: \mathbb{R} \to 
\B$ is the Heaviside function defined as $\mathbb{I}(x) = 0$ if $x \leq 0$ and $1$ 
otherwise. A very relevant result obtained by Mendoza and Alvarez-Buylla was that the 
dynamical behavior of this network, according to a specific block-sequential (or 
series-parallel) updating mode (see~\cite{Robert1986} for more details about this 
kind of updating mode), converges towards six fixed points among which four 
corresponds to the four floral cellular types: carpels, stamens, petals and sepals. 
The two other fixed points are respectively related to inflorescence and mutant 
cellular types.
												
\begin{figure}[t!]
	\centerline{\scalebox{.85}{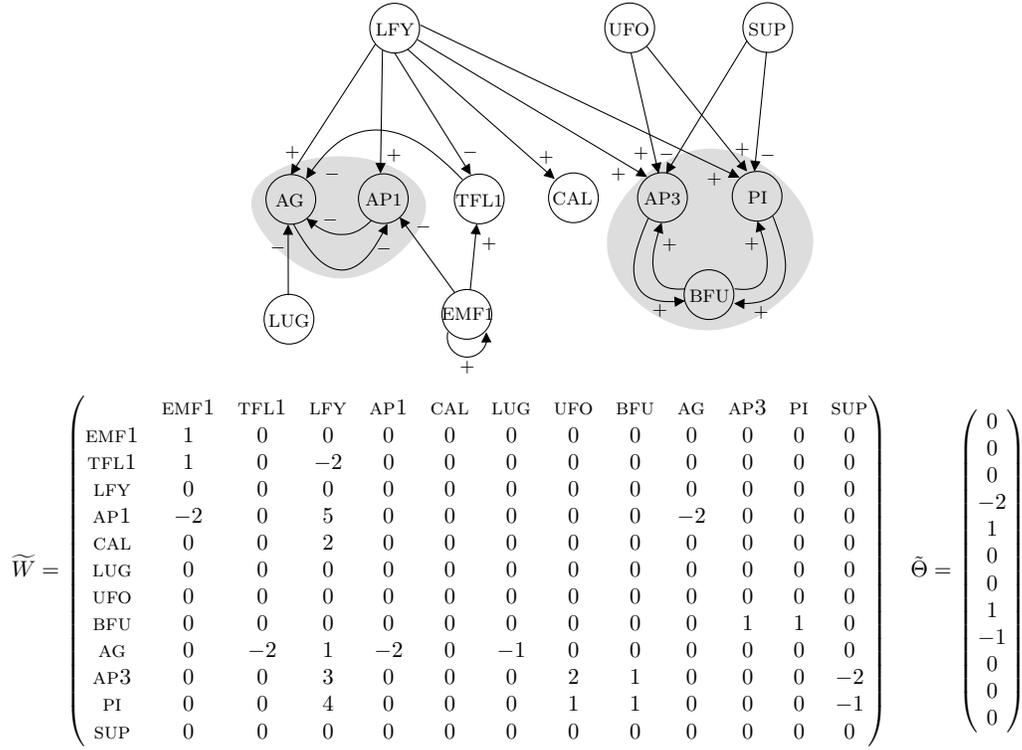}}
	\caption{Symmetric version of the regulation network modeling the genetic control 
		of \emph{Arabisopsis thaliana} floral morphogenesis pointing out the dynamical 
		role of the two strongly connected components $\mathcal{C}_1$ and 
		$\mathcal{C}_2$.}
	\label{fig:mendozaSym}
\end{figure}

In~\cite{Demongeot2010}, the authors emphasized that this network, when subjected to 
the parallel updating mode, has also seven limit cycles of size $2$. Basing 
themselves on the results obtained by Goles in~\cite{Goles1980} about the convergence 
of symmetric threshold BNs, they showed that this original network was actually 
equivalent to another simpler one. The latter is simpler in the sense that its 
dynamical richness is entirely governed by two strongly connected components 
$\mathcal{C}_1$ and $\mathcal{C}_2$, of respective sizes $2$ and $3$ and such that 
$V_{\mathcal{C}_1} = \{\gAG, \gAPI\}$ and $V_{\mathcal{C}_2} = \{\gAPIII, \gBFU, 
\gPI\}$. This network, together with its related interaction matrix $\tilde{W}$ and 
threshold vector $\tilde{\Theta}$ is pictured in Figure~\ref{fig:mendozaSym}. From 
the analysis of the dynamical behavior of this network given in~\cite{Demongeot2010}, 
we derive Theorem~\ref{thm:Mendoza} below.

\begin{theorem}
	\label{thm:Mendoza}
	Suppose that every node that does not belong to neither $\mathcal{C}_1$ nor 
	$\mathcal{C}_2$ is at state $0$. The following holds:
	\begin{enumerate}
	\item[i)] Any MBN associated with $\mathcal{C}_1$ admits the two fixed 
		points $(dt_{\gAG}, 0)$ and $(0, dt_{\gAPI})$. Only those such that 
		$dt_{\gAG} = dt_{\gAPI}$ can admit a limit cycle.
	\item[ii)] For the $C_2$ component, the behavior of any MBN built on it is 
		the following:
		\begin{itemize}
		\item If $d_{\gBFU} \geq 2$, then the set of attractors of the network is 
			composed of two fixed points, $(0,0,0)$ and $(d_{\gAPIII}, d_{\gBFU}, 
			d_{\gPI})$.
		\item If $(dt_{\gBFU} = 1) \land (dt_{\gAPIII} \geq 2) \land (dt_{\gPI} \geq 
			2)$ then the set of attractors of the network is composed of two fixed 
			points, $(0,0,0)$ and $(dt_{\gAPIII}, dt_{\gBFU}, d_{\gPI})$.
		\item If $(d_{\gBFU} = 1) \land ((d_{\gAPIII} = 1) \lor (d_{\gPI} = 1))$) then  
			the set of attractors of the network is composed of two fixed points, 
			$(0,0,0)$ and $(dt_{\gAPIII}, dt_{\gBFU}, d_{\gPI})$, and of a limit 
			cycle of length $2$.
	 	\end{itemize}
	\end{enumerate}
\end{theorem}
	
\begin{proof}
	Let us work on both sub-networks.
	\begin{enumerate}
	\item[\emph{i)}] The proof is the same as that of Theorem~\ref{thm:Lambda}.
	\item[\emph{ii)}] Let us denote by $(\rho, \delta, \gamma)$ the initial 
		configuration. First, notice that whatever the delay vector, if the state of 
		$\gBFU$ is $0$ and at least one of the other two is also $0$ in the initial 
		configuration then the latter converges towards the fixed point $(0,0,0)$. 
		Indeed, $(0,0,\gamma) \to (0,0,\gamma-1) \xrightarrow{\gamma-1} 
		(0,0,0)$\loopr{} that is reached in $\gamma$ time steps, \ie in at most 
		$dt_\gPI$ time steps. Similarly, $(\rho,0,0) \to (\rho-1,0,0) 
		\xrightarrow{\rho-1} (0,0,0)$\loopr{} that is reached in $\gamma$ time steps, 
		\ie in at most $dt_\gPI$ time steps.
		
		Consider now that it is not the case so that the central site and one of its 
		neighbors are not both at state $0$ at the same time. Three cases are 
		possible:
		\begin{enumerate}
		\item $dt_\gBFU \geq 2$. Then, for $1 \leq \delta \leq dt_\gBFU$ and for 
			any $\rho, \gamma \geq 0$ we have:
			\begin{itemize}
			\item if $\rho = 0$ or $\gamma = 0$ then $(\rho, \delta, \gamma) \to 
				(dt_\gAPIII, \delta - 1, dt_\gPI)$. If $\delta-1 = 0$ then 
				\begin{itemize}
				\item if $dt_\gAPIII = dt_\gPI = 1$ then $(dt_\gAPIII, 0, dt_\gPI) 
					\to (0, dt_\gBFU, 0) \to (dt_\gAPIII, dt_\gBFU - 1, dt_\gPI) 
					\to (dt_\gAPIII, dt_\gBFU, dt_\gPI)$\loopr{}.
				\item if $dt_\gAPIII = 1\ \land\ dt_\gPI > 1$ then 
					$(dt_\gAPIII, 0, dt_\gPI) \to (0, dt_\gBFU, dt_\gPI - 1) \to 
						(dt_\gAPIII, dt_\gBFU - 1, dt_\gPI) \to 
						(dt_\gAPIII, dt_\gBFU, dt_\gPI)$\loopr{}.
				\item if $dt_\gAPIII > 1\ \land\ dt_\gPI = 1$ then 
					$(dt_\gAPIII, 0, dt_\gPI) \to (dt_\gAPIII-1, dt_\gBFU, 0) \to 
						(dt_\gAPIII, dt_\gBFU - 1, dt_\gPI) \to 
						(dt_\gAPIII, dt_\gBFU, dt_\gPI)$\loopr{}.
				\item if $dt_\gAPIII > 1\ \land\ dt_\gPI > 1$ then 
					$(dt_\gAPIII, 0, dt_\gPI) \to 
					(dt_\gAPIII - 1, dt_\gBFU, dt_\gPI - 1) \to 
					(dt_\gAPIII, dt_\gBFU, dt_\gPI)$\loopr{}.
				\end{itemize}
				If $\delta-1 \geq 1$ then $(dt_\gAPIII, \delta - 1, dt_\gPI) \to 
				(dt_\gAPIII, dt_\gBFU, dt_\gPI)$\loopr{}.
			\item if $\rho \geq 1$ and $\gamma \geq 1$ then it is trivial that 
				$(\rho, \delta, \gamma)$ converges directly towards 
				$(dt_\gAPIII, dt_\gBFU, dt_\gPI)$.
			\end{itemize}
			Now, if $\rho \geq 1$, $\delta = 0$ and $\gamma \geq 1$ then 
			\begin{itemize}
			\item if $\rho = \gamma = 1$ then $(1, 0, 1) \to (0, dt_\gBFU, 0) \to 
				(dt_\gAPIII, dt_\gBFU - 1, dt_\gPI) \to 
				(dt_\gAPIII, dt_\gBFU, dt_\gPI)$\loopr{}.
			\item if $\rho = 1\ \land\ \gamma \geq 2$ then $(1, 0, \gamma) \to 
				(0, dt_\gBFU, \gamma - 1) \to (dt_\gAPIII, dt_\gBFU - 1, dt_\gPI) \to 
				(dt_\gAPIII, dt_\gBFU, dt_\gPI)$\loopr{}.
			\item if $\rho \geq 2\ \land\ \gamma = 1$ then $(\rho, 0, 1) \to 
				(\rho-1, dt_\gBFU, 0) \to (dt_\gAPIII, dt_\gBFU - 1, dt_\gPI) \to 
				(dt_\gAPIII, dt_\gBFU, dt_\gPI)$\loopr{}.
			\item if $\rho \geq 2\ \land\ \gamma \geq 2$ then $(\rho, 0, \gamma) \to 
				(\rho-1, dt_\gBFU, \gamma-1) \to 
				(dt_\gAPIII, dt_\gBFU, dt_\gPI)$\loopr{}.
			\end{itemize}
			Thus, when $dt_\gBFU \geq 2$, the network admits only two attractors, 
			fixed point $(0,0,0)$ (resp. $(dt_\gAPIII, dt_\gBFU, dt_\gPI)$) towards 
			which it converges in at most $\max(dt_\gAPIII, dt_\gPI)$ (resp. $4$) time 
			steps.

		\item $dt_\gBFU = 1$ and $(dt_\gAPIII \geq 2)$ and $(dt_\gPI \geq 2)$. 
			Given an initial configuration $(\rho, \delta, \gamma)$, we have:
			\begin{itemize}
			\item if $\delta = 0$ and $\rho \geq 1$ and $\gamma \geq 1$ then:
				\begin{itemize}
				\item if $\rho = \gamma = 1$ then $(1, 0, 1) \to (0, 1, 0) \to 
					(dt_\gAPIII, 0, dt_\gPI)$. 
				\item if $\rho = 1\ \land\ \gamma \geq 2$ then $(1, 0, \gamma) \to 
					(0,1,\gamma-1) \to (dt_\gAPIII, 0, dt_\gPI)$.
				\item if $\rho \geq 2\ \land\ \gamma = 1$ then $(\rho, 0, 1) \to 
					(\rho-1,1,0) \to (dt_\gAPIII, 0, dt_\gPI)$.
				\item if $\rho, \gamma \geq 2$ then $(\rho, 0, \gamma) 
					\to (\rho-1,1,\gamma-1) \to (dt_\gAPIII, 0, dt_\gPI)$.
				\end{itemize}
				Now, notice that $(dt_\gAPIII, 0, dt_\gPI) \to 
				(dt_\gAPIII-1, 1, dt_\gPI-1) \to 
				(dt_\gAPIII, dt_\gBFU, dt_\gPI)$\loopr{}.
			\item if $\delta = 1$ and $\rho \geq 1$ and $\gamma \geq 1$ then it is 
				trivial that there is a direct convergence towards 
				$(dt_\gAPIII, dt_\gBFU, dt_\gPI)$. 
			\end{itemize}
			Thus, when $dt_\gBFU = 1$ and $dt_\gAPIII \geq 2$ and $(dt_\gPI \geq 2)$, 
			the network necessarily converges towards either $(0,0,0)$ in at most 
			$\max(dt_\gAPIII, dt_\gPI)$ time steps or 
			$(dt_\gAPIII, dt_\gBFU, dt_\gPI)$ in at most $4$ time steps.

		\item $dt_\gBFU = 1$ and $((dt_\gAPIII = 1)\ \lor\ (dt_\gPI = 1))$. Let us 
			first suppose that $dt_\gAPIII = 1$ and $dt_\gAPIII \geq 2$ and consider 
			an initial configuration $(\rho, \delta, \gamma)$. We have:
			\begin{itemize}
			\item if $\delta = 0$ and $\rho = 1$ and $\gamma \geq 1$ then 
				$(1, 0, \gamma) \to (0, 1, \gamma-1) \to (dt_\gAPIII, 0, dt_\gPI) 
				\leftrightarrows (0, 1, dt_\gPI-1)$, which highlights a limit cycle of 
				length $2$.
			\item if $\delta = 1$ then:
				\begin{itemize}
				\item $\forall \rho, \gamma$ such that $\rho = 0$ or $\gamma = 0$, 
					$(\rho,1,\gamma) \to (dt_\gAPIII, 0, dt_\gPI) \leftrightarrows 
					(0, 1, dt_\gPI-1)$, which highlights a limit cycle of length $2$.
				\item if $\rho = 1$ and $\gamma \geq 1$ then $(1,1,\gamma) \to 
					(dt_\gAPIII, dt_\gBFU, dt_\gPI)$\loopr{}.
				\end{itemize}
			\end{itemize}
			The cases where $(dt_\gAPIII \geq 2\ \land\ dt_\gAPIII = 1)$, and where 
			$dt_\gAPIII = dt_\gAPIII = 1$ can be treated similarly.
	 	\end{enumerate}			
	\end{enumerate}
\end{proof}
																				
Theorem~\ref{thm:Mendoza} shows that both components $\mathcal{C}_1$ and 
$\mathcal{C}_2$ that are the engines of the dynamical behavior of the network admit a 
limit cycle. However, here again, these two limit cycles are spurious in the sense 
that only very specific initial conditions allow to capture them. Furthermore, both 
of these components admit two fixed points: $(x_\gAG=dt_\gAG, x_\gAPI=0)$ and 
$(x_\gAG=0, x_\gAPI=dt_\gAPI)$ for $\mathcal{C}_1$, and $(x_\gAPIII=0, x_\gBFU=0, 
x_\gPI=0)$ and $(x_\gAPIII=dt_\gAPIII, x_\gBFU=dt_\gBFU, x_\gPI=dt_\gPI)$. By 
combining them by making vectors $(x_\gAG, x_\gAPI, x_\gAPIII, x_\gBFU, x_\gPI)$, we 
obtain four possible fixed points: 
\begin{gather*}
	\fp_1 = (dt_\gAG, 0, 0, 0, 0),\quad \fp_2 = (0, dt_\gAPI, 0, 0, 0),\\
	\fp_3 = (dt_\gAG, 0, dt_\gAPIII, dt_\gBFU, dt_\gPI) \text{ and } 
	\fp_4 = (0, dt_\gAPI, dt_\gAPIII, dt_\gBFU, dt_\gPI)\text{,}
\end{gather*} 
that correspond exactly to the floral organs (sepals, petals, carpels and stamens) 
according to the ABC model~\cite{Coen1991,Meyerowitz1994}.

\section{Conclusion and perspectives}
\label{sec:conclusion}
																				
This paper aimed at studying from the theoretical point of view the GPBN model of 
Graudenzi and Serra in~\cite{Graudenzi2010,Graudenzi2011a,Graudenzi2011b}. This 
mathematical model was introduced as an extension of BNs allowing to consider nodes 
representing both genes (as Boolean variables) and proteins (as discrete decay times 
of their concentrations). By using an intermediary model, that we call MBNs, we first 
obtained a result emphasizing that GPBNs are computationally speaking equivalent to 
BNs. Indeed, given an arbitrary GPBN, it is possible to construct an dynamically 
equivalent BN composed of more nodes. This result led us to focus on theoretical 
properties of MBNs. For this model in general, we proved that: \emph{(i)} positive 
disjunctive MBNs, as well as locally monotonic MBNs (\ie MBNs whose global transition 
function is either decreasing or increasing for all $x$), necessarily converge 
towards fixed points; \emph{(ii)} the only attractors that MBNs whose 
underlying interaction graphs do not induce any cycle except possible loops 
may have are fixed points. Then, before we presented applications to two well 
know examples of real genetic regulation networks, we analyzed exhaustively the 
possible dynamical behaviors of MBNs composed of two nodes. The underlying idea aimed 
at obtaining a subtle knowledge about the dynamical and computational properties of 
simple biological patterns. It led us to highlight the conditions under which 
this or that pattern admits only fixed points or limit cycles as limit behaviors. 
Moreover, it led us to show that adding memory to BNs is a pertinent way to freeze 
some cyclic behaviors, which is particularly interesting in the case of spurious 
complex attractors, but that it is also a way to create complex attractors in very 
specific networks having few attractors in their initial Boolean form. The features 
of such BNs would deserve to be studied more generally.\medskip

This work opens the following perspectives:
\begin{itemize}
\item \emph{about other applications}: One of the first perspectives would be to 
	use MBNs to model other networks known to model genetic controls in living 
	systems. Among the networks that we think of are the embryonic segmentation of 
	\emph{Drosophila melanogaster}~\cite{Albert2003}, B-cell 
	differentiation~\cite{Mendoza2016} and the fission yeast cell cycle 
	network~\cite{Li2004}.
\item \emph{about larger networks}: A natural opening of this work is to focus on the 
	theoretical properties of the dynamics and the complexity of networks of larger 
	sizes. A first axis would be to maintain our interest in biological patterns by 
	paying attention to networks of size $3$ and $4$. Because of the intrinsic 
	complexity (regarding the size of the phase space) of such a work, it would be 
	judicious to focus on specific patterns that have been highlighted to be either 
	statistically well represented in biological networks~\cite{Alon2002} or of 
	specific importance such as regulons or feed-forward coherent and incoherent 
	cycles. Moreover, of course, another axis should naturally be articulated around 
	larger networks whose interaction graphs belong to specific classes, such as 
	cycles and more generally graph family like cacti and caterpillars.
\item \emph{about synchronicity}: As evoked in Remark~\ref{rem:synchro}, MBNs get 
	interesting features to understand the role of synchronicity and asynchonicity 
	on the dynamics of regulatory networks. Indeed, when we work on BNs, the way to 
	study and understand the influence of synchronicity is to play with the updating 
	mode (parallel, block-sequential, sequential, randomly sequential, fair, 
	asynchronous...). However, although there exist a lot of results in this domain, 
	the comparison of their impact on networks is very tricky for mainly two reasons: 
	the infinite number of possible updating modes in theory and their mathematical 
	nature (\ie deterministic vs non-deterministic, periodic vs non-periodic...). 
	Here, with MBNs, while we keep the parallel updating mode, we can directly change 
	synchronicity through initial configurations that play the role of synchronicity 
	control parameters. From the biological point of view, this is of peculiar 
	interest because in most cases, the synchronicity induced by the use of the 
	parallel updating mode (that the most basic mathematically speaking) tends to 
	produce cyclic attractors with no biological meaning. In this framework, future 
	works will be oriented towards the characterizations of delays for distinct 
	classes of MBNs ensuring to filter such spurious attractors by freezing the 
	networks. 
\end{itemize}\medskip

\paragraph{Acknowledgements} This work has been supported by ECOS-CONICYT C16E01 (EG, 
  FL, GR, SS), FONDECYT 1140090 (EG, FL), Basal Project CMM (GR, EG), FANs program 
  ANR-18-CE40-0002-01 (SS) and PACA Fri Project 2015\_01134 (SS).

\bibliographystyle{plain}
\bibliography{paper_GLRS_NACO_final}   

\end{document}